\documentclass[12pt]{article}
\usepackage{amsmath}
\usepackage{graphicx,psfrag,epsf}
\usepackage{enumerate}
\usepackage{natbib}
\usepackage{url} 
\usepackage{bm}         
\usepackage{amsmath}    
\usepackage{amsfonts}   
\usepackage{amssymb}    
\usepackage{multirow}
\usepackage{amsthm}
\usepackage{morefloats}
\usepackage{subcaption}
\usepackage{caption}

\def\bSig\mathbf{\Sigma}
\usepackage{bm}         
\usepackage{amsmath}    
\usepackage{amsfonts}   
\usepackage{amssymb}    
\def\bSig\mathbf{\Sigma}

\newcommand{\tr}{\mbox{tr}}
\newcommand{\eig}{\mathrm{eig}}
\newcommand{\mcD}{\mathcal{D}}
\newcommand{\mcL}{\mathcal{L}}
\newcommand{\mcW}{\mathcal{W}}

\newcommand{\bX}{\bm{X}}
\newcommand{\bx}{\bm{x}}

\newcommand{\by}{\bm{y}}

\newcommand{\bgamma}{\bm{\gamma}}

\newcommand{\RR}{{\mathbb R}}
\newcommand{\Ind}{1\!\mathrm{l}}

\numberwithin{equation}{section}
\newtheorem{theorem}{Theorem}[section]
\newtheorem{lemma}[theorem]{Lemma}
\newtheorem{proposition}[theorem]{Proposition}

\newtheorem{remark}[theorem]{Remark}

\usepackage{mathtools}

\DeclarePairedDelimiter\floor{\lfloor}{\rfloor}

\usepackage[figuresright]{rotating}

\usepackage{todonotes}

\newcommand{\blind}{0}

\newcommand{\beginsupplement}{%
        \setcounter{table}{0}
        \renewcommand{\thetable}{S\arabic{table}}%
        \setcounter{section}{0}%
        \renewcommand{\thesection}{S\arabic{section}}  
        \setcounter{figure}{0}
        \renewcommand{\thefigure}{S\arabic{figure}}%
     }

\date{}
\begin{document}

\def\spacingset#1{\renewcommand{\baselinestretch}%
{#1}\small\normalsize} \spacingset{1}

%%%%%%%%%%%%%%%%%%%%%%%%%%%%%%%%%%%%%%%%%%%%%%%%%%%%%%%%%%%%%%%%%%%%%%%%%%%%%%

\if0\blind
{
  \title{\bf Bayesian Nonparametric Graph Clustering}
  \author{Sayantan Banerjee\\
    Department of Biostatistics,\\The University of Texas MD Anderson Cancer Center\\
    and \\
    Rehan Akbani\\
    Department of Bioinformatics and Computational Biology, \\The University of Texas MD Anderson Cancer Center\\
    and\\
    Veerabhadran Baladandayuthapani \footnote{Corresponding author, e-mail: veera@mdanderson.org}\\
    Department of Biostatistics,\\ The University of Texas MD Anderson Cancer Center}
  \maketitle
} \fi

\if1\blind
{
  \bigskip
  \bigskip
  \bigskip
  \begin{center}
    {\LARGE\bf Title}
\end{center}
  \medskip
} \fi

\bigskip
\begin{abstract}
We present clustering methods for multivariate data exploiting the underlying geometry of the graphical structure between variables. As opposed to standard approaches that assume known graph structures, we first estimate the edge structure of the unknown graph using  Bayesian neighborhood selection approaches, wherein we account for the uncertainty of graphical structure learning through model-averaged estimates of the suitable parameters. Subsequently, we develop a nonparametric graph clustering model on the lower dimensional projections of the graph based on Laplacian embeddings using Dirichlet process mixture models. In contrast to  standard algorithmic approaches, this fully probabilistic approach allows incorporation of uncertainty in estimation and inference for both graph structure learning and clustering. More importantly, we formalize the arguments for Laplacian embeddings as suitable projections for graph clustering by providing theoretical support for the consistency of the eigenspace of the estimated graph Laplacians. We develop fast computational algorithms that allow our methods to  scale to large number of nodes. Through  extensive simulations we compare our clustering performance with standard clustering methods. We apply our methods to a novel pan-cancer proteomic data set, and evaluate protein networks and clusters across multiple different cancer types.
\end{abstract}

\noindent%
\vfill

\newpage
\spacingset{1.45} 
% DON'T change the spacing!
\section{Introduction}
\label{Sec:Introduction}

Clustering is one of the most widely studied approaches for investigating dependencies in multivariate data that arise in several domains, such as the biomedical and social sciences. Standard clustering approaches are based on metrics that define the similarity and/or dependence between variables. Examples are the use of linkage-based (e.g., hierarchical clustering), distance-based (Euclidean, Manhattan, Minkowski) or kernel-based metrics. In this article, we focus on a particular sub-class of clustering methods that are based on graphical dependencies between the variables, especially for data lying on a graph, or data that can be modeled probabilistically using a graphical model. In many complex systems, graphs or networks can effectively represent the inter-dependence between the major variables of interest in various modeling contexts. Examples include protein-protein interaction networks in various cancer types, which may be modeled statistically to reveal the groups of proteins that play significant roles in different disease etiologies; social networks for users or communities that share common interests; and image networks, which help to identify similarly classified images. 

An important challenge of the investigation, when the number of nodes in a graph or network increases significantly, is to develop computationally efficient approaches for learning the network structure as well as inferring from the same, especially identifying the underlying clusters within the network. The investigation can also be addressed by identifying clusters of objects in the entire network, or as a graph partitioning problem, where the graph is partitioned into clusters such that the within-cluster connections are high and the between-cluster connections are low. Clusters of objects often help to reveal the underlying mechanism of the objects in the network with respect to the problem of interest. The graph partitioning problem is quite prevalent in the computer science and machine learning literature, where the main focus is on partitioning a given graph with known edge structures [see \cite{von2007tutorial} and a comprehensive list of references therein].
 
From a modeling standpoint, graphical models \citep{lauritzen1996graphical} are able to capture the conditional dependence structure through parametric (usually Gaussian) formulations on covariance matrices or precision (inverse covariance) matrices. Bayesian methods of inference that use graphical models have received a great deal of attention recently. One approach uses prior specifications on the space of sparse precision matrices. A conjugate family of priors, known as the $G$-Wishart prior \citep{roverato2000cholesky} has been developed for incomplete decomposable graphs. A more general family of conjugate priors for the precision matrix is the $W_{P_G}$-Wishart family of distributions; see \cite{letac2007wishart, rajaratnam2008flexible, banerjee2014posterior}.  Bayesian methods of inference that use priors on elements of the precision matrix have also been proposed, e.g., Bayesian graphical lasso \citep{wang2012bayesian}, shrinkage priors for precision matrices \citep{khondker2013bayesian, wang2013class}, and Bayesian graphical structure learning \citep{banerjee2015bayesian, wang2015scaling}. An alternative approach is to use regression models to estimate the precision matrix, including neighborhood selection of the vertices of a graph (see \cite{peng2009partial, meinshausen2006high}). Recent work by \cite{kundu2014} estimates the graphical model by using a post-model fitting neighborhood selection approach that is motivated by variable selection using penalized credible regions in the regression set-up, as introduced in \cite{bondell2012consistent}. Precision matrix estimation using Bayesian regression methods has been studied by \cite{dobra2004sparse, bhadra2013joint}, where priors have been used on the regression coefficients, inducing sparsity.

The above methods focus solely on graph structure learning. In this article, we focus on a broader problem of first learning the structure of the underlying graphical model and then learning the clusters using a valid probability model for the graph and the clusters within the graph. When the graphical structure is {\it{known}}, widely used algorithmic  graph partitioning/cutting methods can be used, including spectral graph clustering methods \citep{shi2000normalized, ng2002spectral, von2007tutorial}. Spectral methods serve to explore the geometry of the graphical structure and learn the underlying clusters by using the eigenvectors of suitable matrices (graph Laplacians) obtained from the graph.  These methods originated from primary works of \cite{fiedler1973algebraic} and \cite{donath1973lower}. While these methods have attractive computational properties, the main drawbacks are that they assume known graph structures and, due to the heuristic/algorithmic nature of the formulation, do not allow for the explicit incorporation of uncertainty in estimation and inference for both graph structure learning and clustering. For {\it unknown} graphical structures, a suitable measure of pairwise similarities between the nodes is defined, which is then used to derive the graph Laplacian. This relates to nonlinear dimension reduction techniques or manifold learning such as diffusion maps \citep{coifman2006diffusion} and Laplacian eigenmaps \citep{belkin2003laplacian}. The eigenspace of the graph Laplacian can identify the connected components of the graph; hence, it is important to use a suitable graph Laplacian that is close to the `true' graph Laplacian of the underlying graph. Asymptotic convergence of the graph Laplacian has been well studied in the literature for random graphs, that is, for graphs in which the data points are random samples from a suitable probability distribution [see \citet{rohe2011spectral} for a comprehensive list of references]. However, we deal with graphical models in which the edges do not follow a particular probability distribution, but are defined through suitable strength of association between different variables. In our case, such associations can be effectively expressed through partial correlations so that the conditional dependence structure of the corresponding graphical model is well preserved and interpretable.

In this article, we propose a fully probabilistic approach to the problem. Specifically,  we consider a $p$-dimensional Gaussian random variable $\bX = (X_1,\ldots, X_p)^T$, where the $p$ components of the random variable can be separated into $K \ll p$ (unique) clusters. The primary inferential aim is to retrieve the clusters along with their (connected) components from a random sample of size $n$ from the underlying distribution, where the dimension $p$ may grow with $n$. We de-convolve the problem into two sub-problems. First, we estimate the edge structure of the unknown graph using a Bayesian neighborhood selection approach, wherein we account for the uncertainty of graphical structure learning through model-averaged estimates of the suitable parameters. Second, conditional on the model-averaged estimate, we develop a nonparametric graph clustering on the Laplacian embedding of the variables using Dirichlet process mixture models. More importantly, we formalize the arguments for Laplacian embedding as suitable projections for graph clustering by providing theoretical support. Specifically, we establish the consistency of the eigenspace of the graph Laplacian obtained from the estimated graph, which guarantees that the estimated graph Laplacian can be used as a suitable graph clustering tool. From a computational standpoint, we adopt fast computational methods for both graphical structure learning and clustering of variables using nonparametric Bayesian methods. For clustering in particular, we resort to fast and scalable algorithms based on asymptotic representations of Dirichlet process mixture models as an alternative to computationally expensive posterior sampling algorithms, especially for high-dimensional settings. As an extension of this work, we show that our formulation can be used for multiple data sources, to simultaneously cluster the variables locally for each data set and also arrive at a global clustering for all the data sets combined.

We evaluate the operating characteristics of our methods using both simulated and real data sets. In simulations, we compare our Bayesian nonparametric graph-based clustering method with both standard (``off-the-shelf") graph-based clustering methods as well as methods that  do not incorporate the structural information obtained from a graphical model formulation so as to evaluate the effectiveness of using a graph-based clustering model. Using various metrics of clustering efficiency (such as normalized mutual information scores and between-cluster edge densities), we demonstrate the superior performance of our methods compared to the performance of the competing methods. Our methods are motivated by and applied to a novel pan-cancer proteomic data set, through which we evaluate protein networks and clusters across 11 different cancer types. Our analyses reveals several biologically-driven clusters that are are conserved across multiple cancers as well differential clusters that are cancer-specific.

The paper is organized as follows. In Section~\ref{Sec:Inference}, we lay out the inferential problem and present preliminaries on graphical models and related concepts, including graph Laplacian and Laplacian embedding. In Section~\ref{Sec:Weighted adjacency}, we propose  methods for constructing the weighted adjacency matrices, and in  Section~\ref{Sec:Nonparam graph clust}, we present our nonparametric graph clustering method. We discuss the theoretical results pertaining to the consistency of the graph Laplacian in Section~\ref{Sec:consistency}. We present the results of our simulation study in Section~\ref{Sec:Simulation}, and those of the real data analyses on pan-cancer proteomic networks in Section~\ref{Sec:Real data}. The Supplementary Materials contain extensions of our approach to multiple data sets and the results of additional  simulations and real data analyses. We provide proofs of the theoretical results in an Appendix.

\section{Basic inferential problem and preliminaries}
\label{Sec:Inference}

Let $\bX^{(n)} = (\bX_1,\ldots,\bX_n)$ be independent and identically distributed random variables of dimension $p$ with mean $\bm{0}$ and covariance matrix $\Sigma$. We write $\bX_i = (X_{1,i},\ldots,X_{p,i})^T$, and assume that $\bX_i,\,i=1,\ldots,n$ are multivariate Gaussian variables. 

We also assume that there are $K$ (true) underlying clusters for the $p$ variables defined as $\mathcal{C} = (C_1,\ldots,C_{K})$, which are tightly connected. Our basic inferential problem is to retrieve the clusters of variables from the data, $P(\mathcal{C}|\bX)$, which we do by breaking the problem into two sub-problems. (1) First, we learn the underlying graphical structure:  $P(\mathcal{G}|X)$ (see Section \ref{Sec:Weighted adjacency}). (2) Then, we infer $P(\mathcal{C}|\mathcal{G})$ using the appropriate probability models on the graph structures (see Section \ref{Sec:Nonparam graph clust}). We  discuss each of these constructions in the ensuing sections.

\subsection{Graph Laplacians}
Consider an undirected graph $\mathcal{G} = (V,E)$, where $V$ represents the set of vertices $\{1,\ldots,p\}$ and $E$ is the edge set such that $E \subset \{(l,j) \in V \times V: l < j \}$. Two vertices $v_l$ and $v_j$ are adjacent if there is an edge between them. The adjacency matrix is given by $A = (\!(a_{lj})\!)$, where $a_{lj} = 1$ or $0$ according to whether $(l,j) \in E$ or not. Alternatively, one may also consider a weighted graph, with weighted adjacency matrix $W = (\!(w_{lj})\!)$, where $w_{lj} (\geq 0)$ denotes the ``strength"  of association (defined suitably) and  absence of an edge is reflected by a strictly zero edge weight. The degree of a vertex $v_l$ is given by $d_l = \sum_{j = 1}^p w_{lj}$ and the  degree matrix can then be denoted by $D = \mathrm{diag}(d_1,\ldots,d_p)$. The clusters in the underlying variables are reflected by their conditional independence structure, so that variables in two different clusters are conditionally independent. The underlying graphical model forms a Markov random field such that the absence of an edge indicates that the corresponding random variables are conditionally independent. 

For an undirected weighted graph $\mathcal{G}$ with weighted adjacency matrix $W$ and degree matrix $D$, the graph Laplacian associated with $\mathcal{G}$ is defined as
\begin{equation*}
L = D - W.
\end{equation*}
To be more precise, $L$ is referred to as the unnormalized graph Laplacian. The normalized versions of the graph Laplacians are given by 
\begin{eqnarray}
L_{\mathrm{sym}} &=& I - D^{-1/2}WD^{-1/2}, \nonumber \\ 
L_{\mathrm{rw}} &=& I - D^{-1}W. \nonumber
\end{eqnarray}

All of the graph Laplacians, $L, L_{\mathrm{sym}}$, and $L_{\mathrm{rw}}$, satisfy the property that they are positive semi-definite with $p$ non-negative real valued eigenvalues $0 = \lambda_1\leq \lambda_2 \leq \ldots \leq \lambda_p$. In fact, the multiplicity of the zero eigenvalue of the Laplacian is related to the number of connected components of the underlying graph, which is formally stated in the following proposition.

\begin{proposition}
\label{Prop:Lap}
\emph{Let $\mathcal{G}$ be an undirected graph with non-negative weights. Then the multiplicity $d$ of the zero eigenvalue of the graph Laplacian $L$ (or $L_{\mathrm{rw}}$, $L_{\mathrm{sym}}$) equals the number of connected components in the graph.}
\end{proposition}

A formal proof of the above proposition along with the other properties of graph Laplacians can be found in \cite{von2007tutorial}. In light of the above proposition, graph Laplacians can provide a natural tool for clustering the underlying variables. We present two stylized examples in Figure~\ref{fig:connectLap}, one for a 10-dimensional graph with 3 connected components and another for an 18-dimensional graph with 6 connected components. The eigenvalues of the graph Laplacian equal the number of connected components in both cases. Also, the plot of the eigenvectors corresponding to the zero eigenvalues reveals that the eigenspaces of zero eigenvalues are spanned by indicator vectors that correspond to the connected components.

In our setting, we utilize this information to construct a weighted adjacency matrix of dimension $p$, using suitable edge weights (which we discuss later), and then obtain the corresponding graph Laplacian. For the sake of brevity, we refer to the different forms of the graph Laplacian as $L$, irrespective of whether they are normalized or not.

\begin{figure} 
\centering 
\begin{subfigure}[b]{0.3\textwidth}
\includegraphics[width=\textwidth]{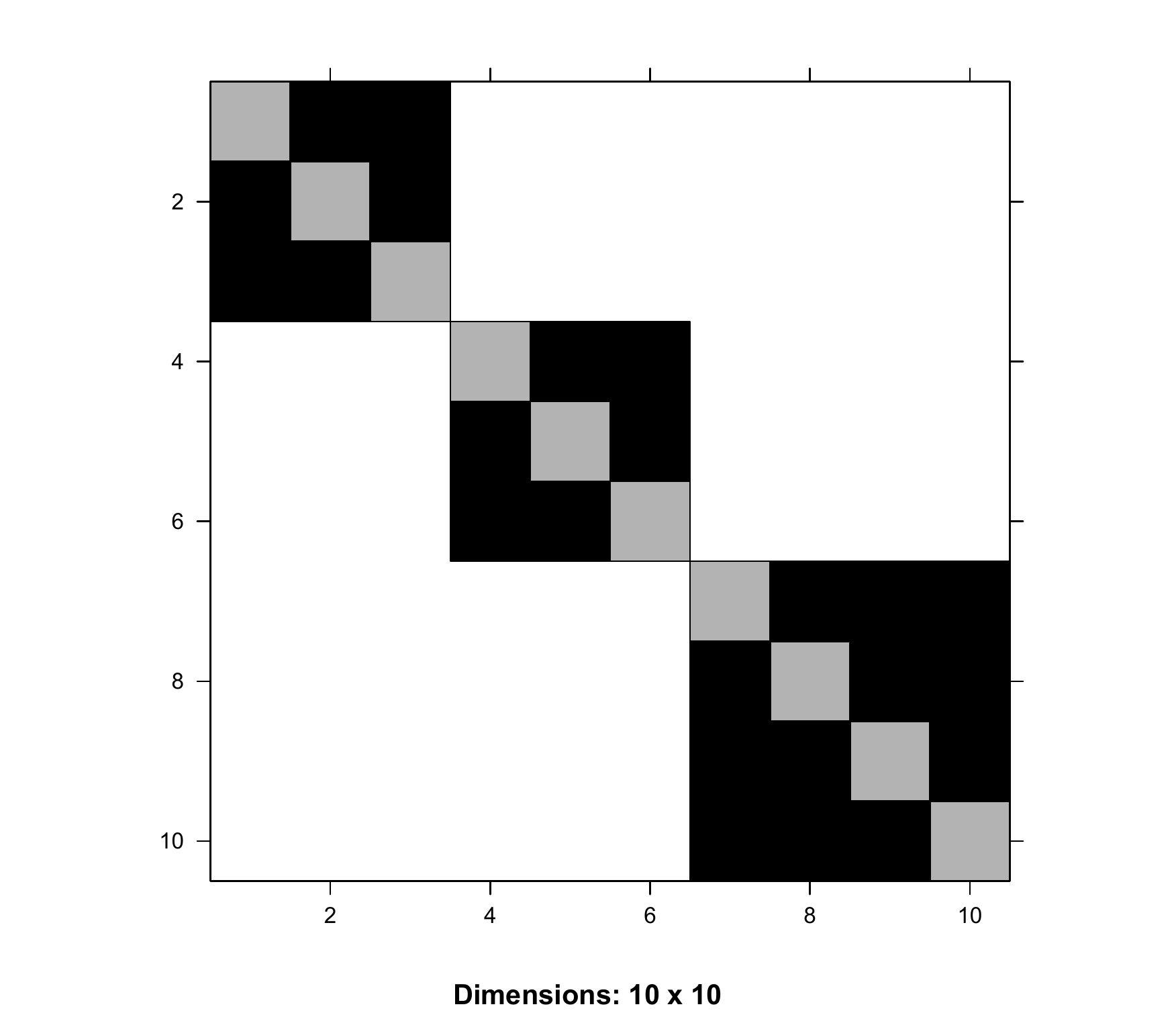} 
\caption{Adjacency matrix} 
\label{fig:gull} 
\end{subfigure} ~ 
%add desired spacing between images, e. g. ~, \quad, \qquad, \hfill etc. %(or a blank line to force the subfigure onto a new line)
\begin{subfigure}[b]{0.3\textwidth}
\includegraphics[width=\textwidth]{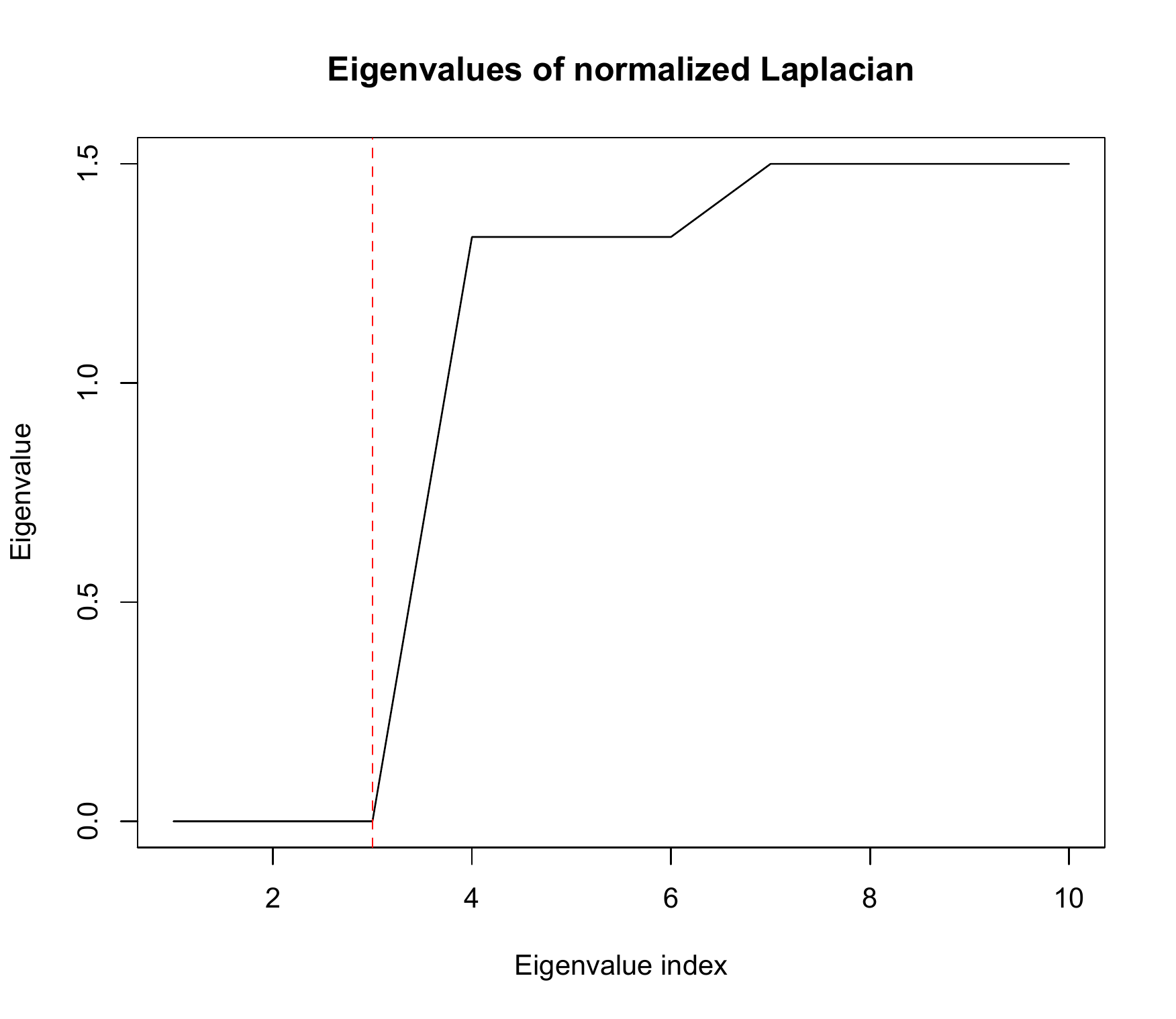} 
\caption{Plot of eigenvalues} \label{fig:tiger} 
\end{subfigure} ~ 
%add desired spacing between images, e. g. ~, \quad, \qquad, \hfill etc. %(or a blank line to force the subfigure onto a new line)
\begin{subfigure}[b]{0.3\textwidth}
\includegraphics[width=\textwidth]{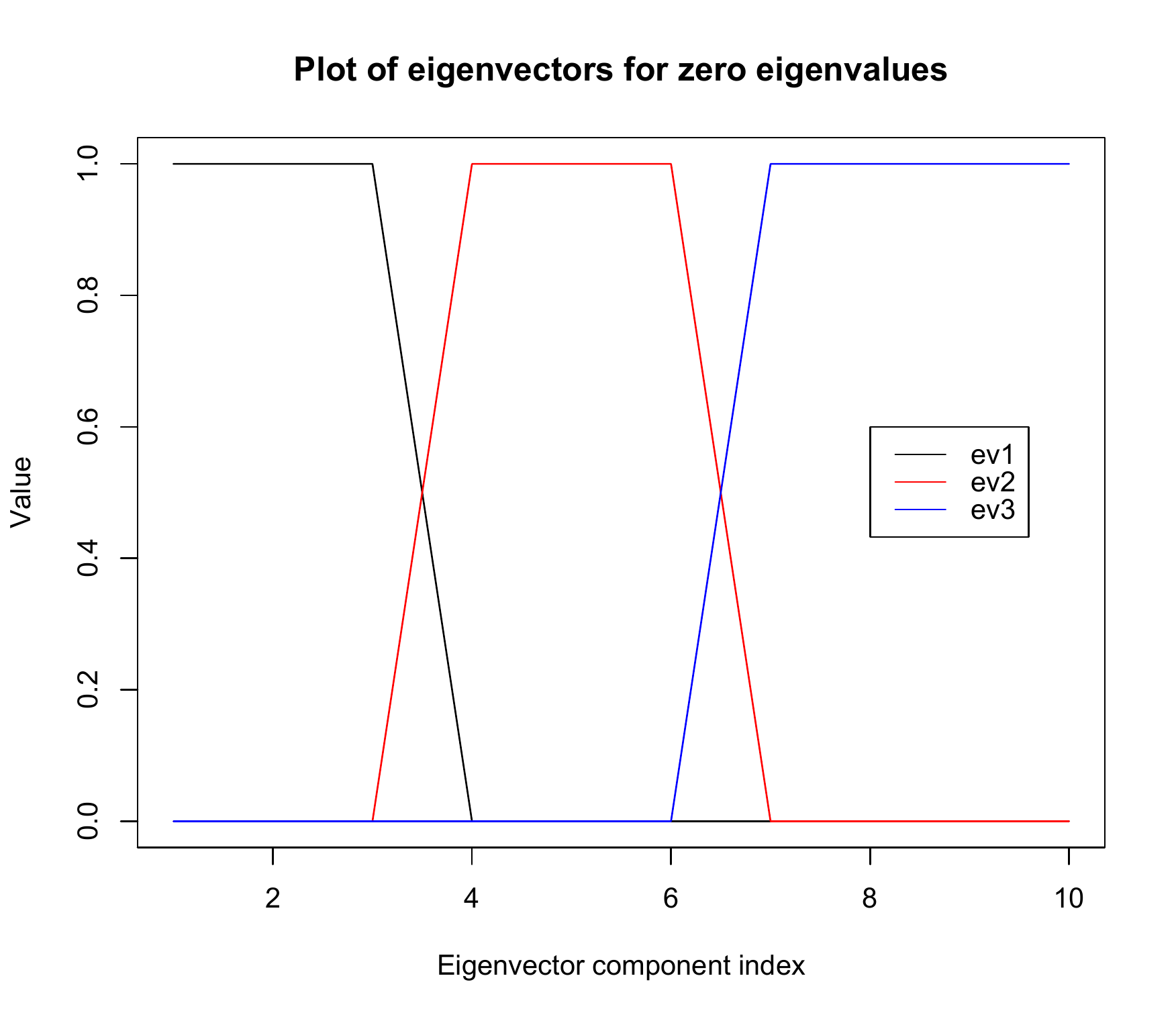} 
\caption{Plot of eigenvectors} 
\end{subfigure} 

\begin{subfigure}[b]{0.3\textwidth}
\includegraphics[width=\textwidth]{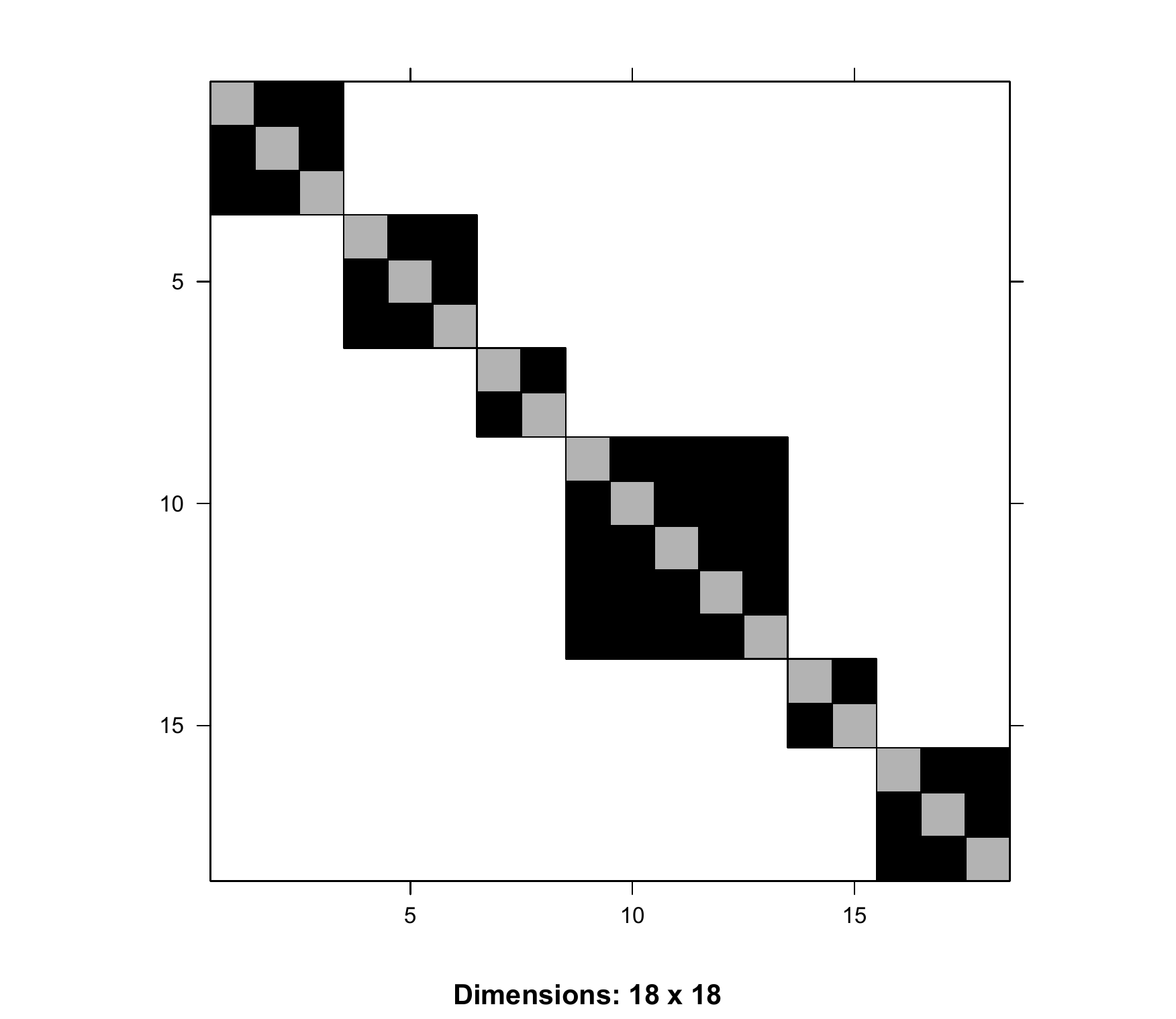} 
\caption{Adjacency matrix} 
\label{fig:gull2} 
\end{subfigure} ~ 
%add desired spacing between images, e. g. ~, \quad, \qquad, \hfill etc. %(or a blank line to force the subfigure onto a new line)
\begin{subfigure}[b]{0.3\textwidth}
\includegraphics[width=\textwidth]{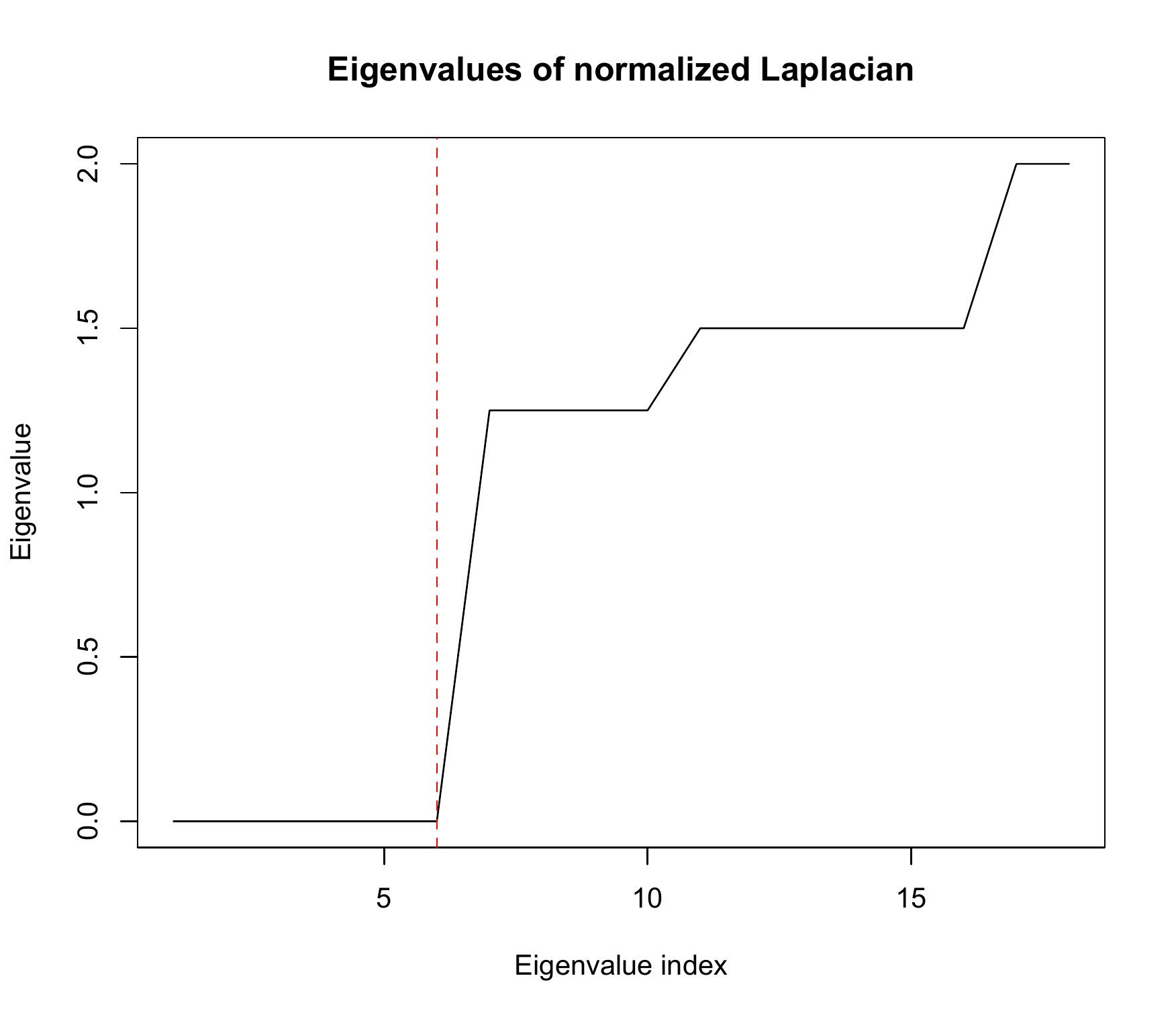} 
\caption{Plot of eigenvalues} \label{fig:tiger2} 
\end{subfigure} ~ 
%add desired spacing between images, e. g. ~, \quad, \qquad, \hfill etc. %(or a blank line to force the subfigure onto a new line)
\begin{subfigure}[b]{0.3\textwidth}
\includegraphics[width=\textwidth]{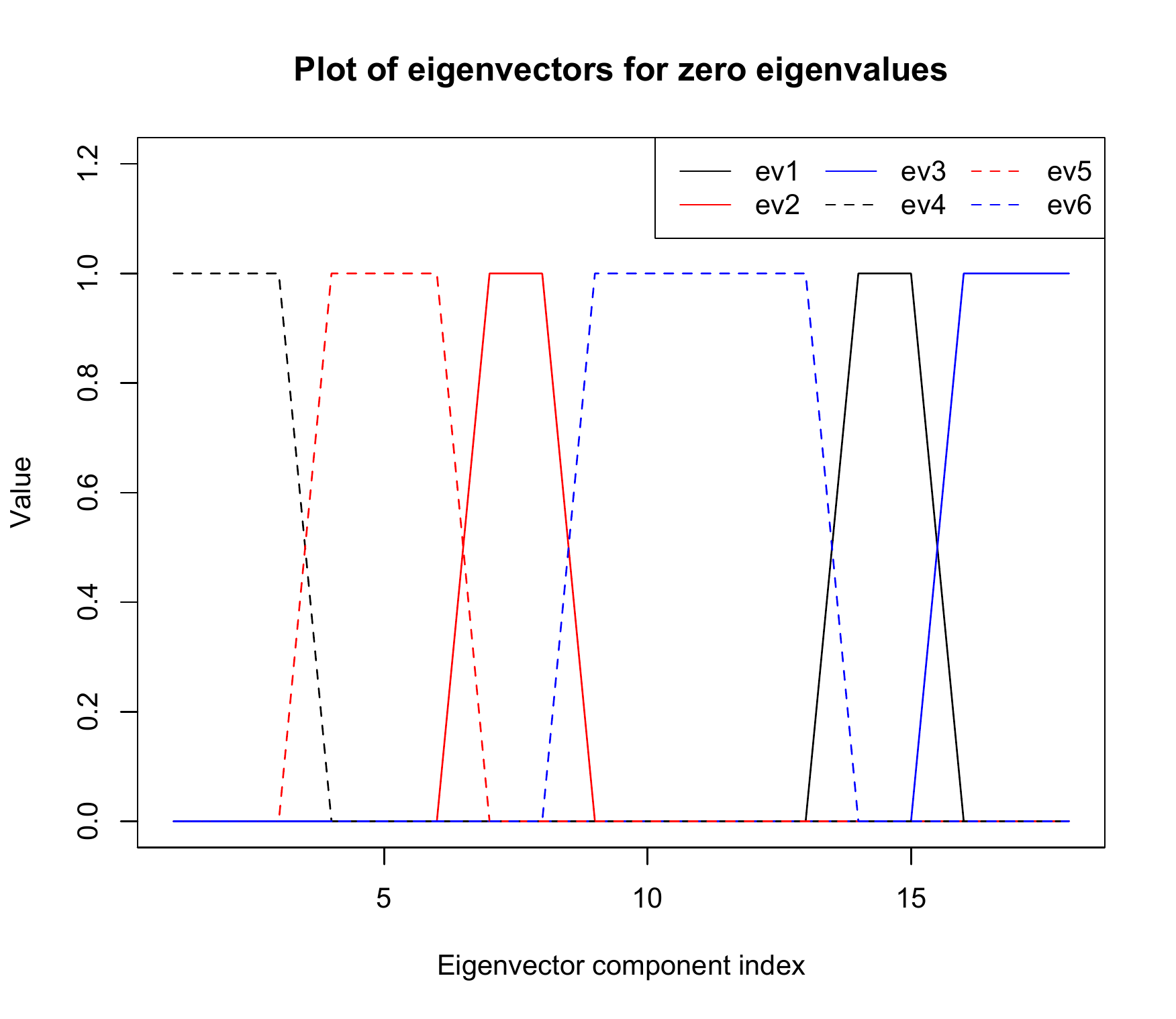} 
\caption{Plot of eigenvectors} 
\end{subfigure} 

\caption{Behavior of graph Laplacians in detecting connected components of a graph. The adjacency matrices are shown in plots (a) a 10-dimensional graph with 3 connected components and (d) an 18-dimensional graph with 6 connected components. Plots (b) and (e) show the eigenvalues of the corresponding normalized Laplacian of the graphs (the red line demarcates the zero and non-zero eigenvalues). Plots (c) and (f) show the eigenvectors (ev) corresponding to the zero eigenvalues of the graph Laplacians.}
\label{fig:connectLap} 
\end{figure}

\subsection{Laplacian embeddings}
Since $L$ is a symmetric matrix, it has an orthonormal basis of eigenvectors. We let $\{ \nu_1, \ldots, \nu_{K_n} \}$ denote the eigenvectors that correspond to the smallest $K_n$ eigenvalues of $L$. We define the normalized Laplacian embedding as the map $\Phi_{\nu}: \{X_1,\ldots,X_p\} \rightarrow \RR^{K_n}$ given by 
\begin{equation} \label{embedding}
\Phi_{\nu}(X_l) = (\nu_{l,1},\ldots,\nu_{l,K_n}).
\end{equation}

Typical spectral clustering methods generally proceed through the following steps. First, the normalized Laplacian is computed and each vertex $X_l,\, l=1,\ldots,p,$ is mapped to a $K_n$-dimensional vector through the above Laplacian embedding. Then, a standard clustering method is applied to the embedded data points. The Laplacian embedding helps to reveal the cluster structure in the data. A formal proof of this phenomenon has been studied recently; see \cite{schiebinger2014geometry}. Thus, if the eigenvalues and eigenvectors of the true graph Laplacian, in other words, the population version of the graph Laplacian, denoted by $\mathcal{L}$, are well approximated by that of the graph Laplacian $L$ obtained from the data set, then Proposition 1.1 ensures that the connected components are revealed by the eigenvalues of $L$, and the corresponding Laplacian embedding helps to find the connected components, hence, the clusters. Accurate approximations of the eigenvectors are ensured under certain conditions on the eigenspace of the underlying matrices, provided $L$ is close to $\mathcal{L}$ in certain matrix norms, such as the operator norm or Frobenius norm (see Section~\ref{Sec:consistency}).

\section{Construction of adjacency matrices}
\label{Sec:Weighted adjacency}

One of the primary steps for graph clustering is the construction of a suitable graph Laplacian for accurate identification of the underlying cluster allocations. Clustering based on Laplacians is closely related to a graph cutting algorithm, which forms partitions on the vertex set in such a way that within-partition edge weights are higher than between-partition edge weights. This requires efficiently and accurately defining the edge weights so as to incorporate this property in the ensuing graph. We focus on Gaussian graphical models in this paper, since they provide a coherent probability model on the space of graphs and are computationally tractable. Specifically, we assume a $p$-dimensional Gaussian random variable $\bX = (X_1,\ldots,X_p)$ having mean $\bm{0}$ and covariance $\Sigma$, where  the precision matrix is given by $\Omega = \Sigma^{-1}$. The corresponding Gaussian graphical model can then be denoted by $\mathcal{G} = (V,E)$, with $V = \{X_1,\ldots,X_p\}$ as the vertex set with an edge present between $X_l$ and $X_j$ if and only of $X_l$ and $X_j$ are conditionally independent given the rest. As before, the absence of an edge is reflected in the entries of $\Omega = (\!(\omega_{lj})\!)$, such that $\omega_{lj} = 0$ if and only if $(l,j) \notin E$.  Thus, the presence or absence of an edge between two variables is reflected in the inverse covariance matrix, and hence in the partial correlation matrix. We can take the absolute value of the partial correlation between $X_l$ and $X_j$ as the edge-weight between $(X_l,X_j)$.

One may also work with an unweighted adjacency matrix that consists only of binary ($0-1$) entries, which might be less efficient in some situations. For example, in cases where the partial correlation between two variables is near zero, assigning an edge between those two does not adequately account for the strength of association between those variables. To overcome this, we use the absolute value of the estimated partial correlations as the edge weights to construct the weighted adjacency matrix $W$ and then derive the graph Laplacian. Estimation of the partial correlations can be performed using empirical methods for low-dimensional problems. For high-dimensional problems such as typical datasets arising from genomic and imaging applications, where $ p > n$, empirical estimates are often unstable, if not estimable. To overcome this instability, sparse methods have been proposed for learning the covariance or inverse covariance matrix, as stated in Section \ref{Sec:Introduction}. In our framework, we want to incorporate the model uncertainty of the graph estimation in the clustering, for which we resort to Bayesian techniques for estimating the precision matrix. Classical variable selection methods typically do not incorporate model uncertainty, as opposed to Bayesian methods, which are capable of providing such measures by using posterior probabilities of models, as detailed below.

\noindent \underline{Bayesian neighborhood selection}:  In this paper, we consider a neighborhood selection approach that uses Bayesian model averaging. Consider $p$ regression models $X_l = \sum_{j \neq l} \beta_{l(j)} X_j + \epsilon_l$,\, $l = 1,\ldots,p$, where $\epsilon_l$ is a zero-mean Gaussian error with variance $\sigma_l^2$ corresponding to the regression of $X_l$ on the rest of the variables. The coefficient estimates $\beta_{l(j)}$ and $\beta_{j(l)}$ obtained from regressing $X_l$ on the rest of the variables and $X_j$ on the rest of the variables can be used to obtain an estimate of the partial correlation between $X_l$ and $X_j$. Note that $\beta_{l(j)} = - \omega_{lj}/\omega_{ll}$ and $\beta_{j(l)} = -\omega_{jl}/\omega_{jj}$. The partial correlation between $X_l$ and $X_j$ is given by $r^{lj} = -\omega_{lj}/(\omega_{ll}\omega_{jj})^{1/2}$, so that $r^{lj}$ may be written as $r^{lj} = \mathrm{sign}(\beta_{l(j)})\sqrt{\beta_{l(j)}\beta_{j(l)}}$. 

We describe the Bayesian model formulation for the above $p$ regression models. We have a $(p-1)$-dimensional regression model in each of the cases with response $X_l$ and regressors $ X_{-l}  := \{X_1,\ldots,X_p\} \backslash \{X_l\},\, l = 1,\ldots, p$. For each of the $p$ models, we define the $(p-1)$-dimensional indicator vector $\bgamma_l = \{\gamma_{l(j)}\}_{j \neq l}$, such that $\gamma_{l(j)} = 1$ if $X_j$ is included in the model and $\gamma_{l(j)} = 0$ otherwise, $l = 1,\ldots,p$. For $l = 1,\ldots, p$, the regression coefficient vector for the regressing $X_l$ on the rest is denoted as $\bm{\beta_l} = \{\beta_{l(j)}\}_{j \neq l}$.

For $l = 1,\ldots, p$, we specify the prior distribution on the model parameters for the $l^\mathrm{th}$ regression model in a hierarchical manner as,
\begin{equation*}
p(\bm{\beta}_{l}, \bgamma_l, \sigma_l^2) = p(\bm{\beta}_{l} \mid  \bgamma_l, \sigma^2)p(\bgamma_l) p(\sigma_l^2).
\end{equation*}
The choice of a prior distribution on the regression coefficients plays an important role in model assessment. In many applications, independent priors are put on the coefficients as a mixture of two components conditioned on $\bgamma_l$, given by
\begin{equation*}
p\{\beta_{l(j)}\} = \{1 - \gamma_{l(j)}\} p^{(0)}\{\beta_{l(j)}\} + \gamma_{(j)} p^{(1)}\{\beta_{l(j)}\}.
\end{equation*}
The `spike and slab' prior \citep{mitchell1988bayesian} chooses the mixture components as $p^{(0)}(\beta_{l(j)}) = \Ind_{\{0\}}(\beta_{l(j)})$ and $p^{(1)}(\beta_{l(j)}) = \Ind_{[-a,a]}(\beta_{l(j)})/2a$ for some $a>0$. \cite{george1993variable} used a different formulation of spike and slab models, assuming that $\bm{\beta}$ has a multivariate Gaussian scale mixture distribution defined by choosing $p^{(0)}(\beta_{l(j)}) = \mathrm{N}(0, \tau_{l(j)}^2)$ and $p^{(1)}(\beta_{l(j)}) = \mathrm{N}(0, c_{l(j)}^2\tau_{l(j)}^2)$. The constant $\tau_{l(j)}^2$ is chosen to be very small so that if $\gamma_{l(j)} = 0$, then $\beta_{l(j)}$ has very negligible variance and may be dropped from the selected model. The constant $c_{l(j)}^2$ is chosen to be large so that if $\gamma_{l(j)} = 1$, then $\beta_{l(j)}$ is included in the final model. The prior for $\gamma_{l(j)}$ is modeled as 
\begin{equation*}
\gamma_{l(j)} \mid u \stackrel{iid}{\sim} (1 - u_{l(j)})\delta_0(\cdot) + u_{l(j)} \delta_1(\cdot), \, u_{l(j)} \stackrel{iid}{\sim} \mathrm{U}(0,1),\, j \neq l.
\end{equation*}
In practice, the selection of the hyperparameters $\tau_{l(j)}^2, c_{l(j)}^2$ and $u_{l(j)}$ may be difficult. Ishwaran and Rao (2000) introduced continuous bimodal priors leading to the hierarchical prior specification as $\beta_{l(j)} \mid \gamma_j, \tau_{l(j)}^2 \stackrel{indep.}{\sim} \mathrm{N}(0,\gamma_{l(j)} \tau_{l(j)}^2), \gamma_{l(j)} \mid u_{l(j)} \stackrel{iid}{\sim} (1 - u_{l(j)})\delta_0(\cdot) + u_{l(j)} \delta_1(\cdot), \tau_{l(j)}^{-2} \mid a_1, a_2 \stackrel{iid}{\sim} \mathrm{Gamma}(a_1,a_2), u_{l(j)} \stackrel{iid}{\sim} \mathrm{U}(0,1).$

The hyperparameters $a_1$ and $a_2$ are chosen so that $\gamma_{l(j)} \tau_{l(j)}^2$ has a continuous bimodal distribution with spike at $0$ and a right continuous tail. However, the effect of the prior vanishes with increasing sample size $n$, for which a rescaled version of the above spike and slab formulation \citep{ishwaran2005spike} was proposed by rescaling the responses by a $\sqrt{n}$-factor , so that $X_{l,i}^* = \sqrt{n}X_{l,i}$, and using a variance inflation factor in the variance of the responses to account for the rescaling. The rescaled spike and slab model is especially useful in high dimensions, and is specified in a hierarchical manner (for the $l$th regression) as,
\begin{eqnarray}
X_{l,i}^* \mid X_{-l,i}, \bm{\beta_l}, \sigma^2 &\stackrel{indep.}{\sim}& \mathrm{N}(X_{-l,i}'\bm{\beta_l}, \sigma_l^2n), \nonumber \\
\beta_{l(j)} \mid \gamma_{l(j)}, \tau_{l(j)}^2 &\stackrel{indep.}{\sim}& \mathrm{N}(0,\gamma_{l(j)} \tau_{l(j)}^2), \nonumber \\
\gamma_{l(j)} \mid u_{l(j)} &\stackrel{iid}{\sim}& (1 - u_{l(j)})\delta_0(\cdot) + u_{l(j)} \delta_1(\cdot),\nonumber \\
\tau_{l(j)}^{-2} \mid a_1, a_2 &\stackrel{iid}{\sim}& \mathrm{Gamma}(a_1,a_2), \nonumber \\
u_{l(j)} &\stackrel{iid}{\sim}& \mathrm{U}(0,1), \nonumber \\
\sigma_l^{-2} &\sim& \mathrm{Gamma}(b_1,b_2).
\end{eqnarray}

\cite{ishwaran2011consistency} proved strong consistency of the posterior mean of the regression coefficients under the rescaled model. We adopt  the fast Gibbs sampling schemes proposed by \cite{ishwaran2005spike} which is detailed in the Supplementary Material (Section S2.1).

We put a rescaled spike and slab prior on the regression coefficients $\beta_{lj}$ for the $p$ regressions in our context, and obtain parameter estimates using Bayesian model averaging. The estimated parameters are then used to obtain the estimated partial correlation matrix $R = (\!(r^{lj})\!)$, using the relationship between the partial correlations and the regression parameters mentioned above. The partial correlation matrix is then used to construct the weighted adjacency matrix $W$ that corresponds to the graphical model. We use Bayesian model averaged estimates of the regression coefficients so that the model uncertainties are propagated while estimating the partial correlation matrix, which are subsequently used to calculate the Laplacian embedding using equation (\ref{embedding}).

\section{Nonparametric graph clustering using Laplacian embeddings}
\label{Sec:Nonparam graph clust}

\subsection{Dirichlet process mixture models}

Spectral clustering methods apply standard clustering techniques to the embedded data points with  a priori fixed number of clusters. A practical way to choose the number of (unknown) clusters $k$ is to identify the $k$ smallest eigenvalues of the graph Laplacian $L$ and then perform standard algorithmic clustering methods such as the k-means method. The major drawback in this regard is the requirement to prespecify the number of clusters, and also that the k-means algorithm tends to select clusters with nearly equal cluster sizes (as we demonstrate in simulations). To avoid such a situation, we propose to use Dirichlet process mixture models (DPMMs), which offer several computational and inferential advantages. We model the embedded data points obtained from the estimated graph Laplacian (as mentioned in the previous section) using a DPMM. 

Specifically, we denote the embedded data points for variable $X_l$ by $\by_l = \Phi_\nu(X_l),\, l = 1,\ldots,p$, where the dimension of the embedding (denoted by $K_n$ earlier) is chosen to be the number of infinitesimally small eigenvalues of the estimated graph Laplacian. Note that, in a finite Gaussian mixture model, the data are assumed to arise from the distribution
\begin{equation*}
p(\by) = \sum_{c=1}^{C} \mathrm{N}(\by \mid \mu_c, \Sigma_c),
\end{equation*} 
where $\pi_c$ are the mixing parameters for $C$ fixed components, and $\mu_c, \Sigma_c$ are the mean and covariance of the corresponding Gaussian mixture components. We further assume that the covariances $\Sigma_c$ are fixed to be $\sigma I$ for all $c = 1,\ldots,C$, $\sigma > 0$. A standard Bayesian approach for inference in the above set-up can be described by the following hierarchical representation, 
\begin{eqnarray*}
\pi_1,\ldots,\pi_C \mid \alpha_0 &\sim& \mathrm{Dir}\left(\alpha_0/C, \ldots, \alpha_0/C\right),\\
z_1,\ldots,z_p \mid \pi_1,\ldots,\pi_C &\sim& \mathrm{Discrete}(\pi_1,\ldots,\pi_C),\\
\mu_c \mid \rho &\sim& G_0(\rho),\\
\by_l \mid z_l, \{\mu_c\}_{c=1}^C & \sim & \mathrm{N}(\mu_{z_l}, \sigma I),
\end{eqnarray*}
for some suitable prior distribution $G_0$. Since the covariances $\Sigma_c$ are fixed, a prior distribution over the means may be chosen to be $\mathrm{N}(\bm{0}, \rho I)$, so that conditional distributions of the parameters may be obtained in a closed form for Gibbs sampling. In the above model, $z_l \in \{1,\ldots,C\},\, l =1,\ldots,p$ serves as the indicator variable for data point $\by_l$ for a specific cluster. One of the $C$ clusters is first chosen according to the multinomial distribution parameterized by $\pi_1,\ldots,\pi_C$, followed by sampling from the corresponding Gaussian distribution parameterized by $\mu_{z_i}$. The mixture weights follow a symmetric Dirichlet distribution with hyperparameter $\alpha_0$. A DPMM is obtained from the same generating process by letting $C \to \infty$, and replacing the Dirichlet prior for $\pi_1,\ldots,\pi_C$ by using a stick-breaking construction \citep{sethuraman}, 
\begin{eqnarray*}
\beta_j &\sim& \mathrm{Beta}(1,\alpha_0),\\
\pi_j &=& \beta_j \prod_{l=1}^{j-1}(1 - \beta_l).
\end{eqnarray*}
The corresponding sequence $\{\pi_j\}_{j=1}^{\infty}$ satisfies $\sum_{j=1}^{\infty} \pi_j = 1$ with probability one.

DPMMs have the intrinsic property of clustering the data so that hard clusterings may be obtained by simulating from the posterior. A concise schematic representation of our  Bayesian nonparametric graph clustering model in shown in Figure \ref{fig:dagmodel}.

\subsection{Posterior computations}

\subsubsection{MCMC}

Conditional posterior distributions of the cluster indicators are utilized for Gibbs sampling, looping through each data point $\by_l, \, l= 1,\ldots,p.$ Data point $\by_l$ is assigned to an existing cluster $c$ with probability $n_{c,-l}\mathrm{N}(\by_l \mid \mu_c, \sigma I)$, where $n_{c,-l}$ is the number of data points in cluster $c$, except $\by_l$. A new cluster is started with probability proportional to $\alpha_0 \int \mathrm{N}(\by_l \mid \mu,\sigma I) dG_0(\mu).$ The conditional posterior distributions of the means are also obtained and used for Gibbs updating given all the data points in a particular cluster, after resampling all the clusters. Detailed procedures for using the Gibbs sampler are discussed in \cite{west1994hierarchical} and \cite{neal2000markov}. For completeness, we present the details of the MCMC procedure in the Supplementary Material Section S2.2.

\begin{figure}
\includegraphics[scale=1.4]{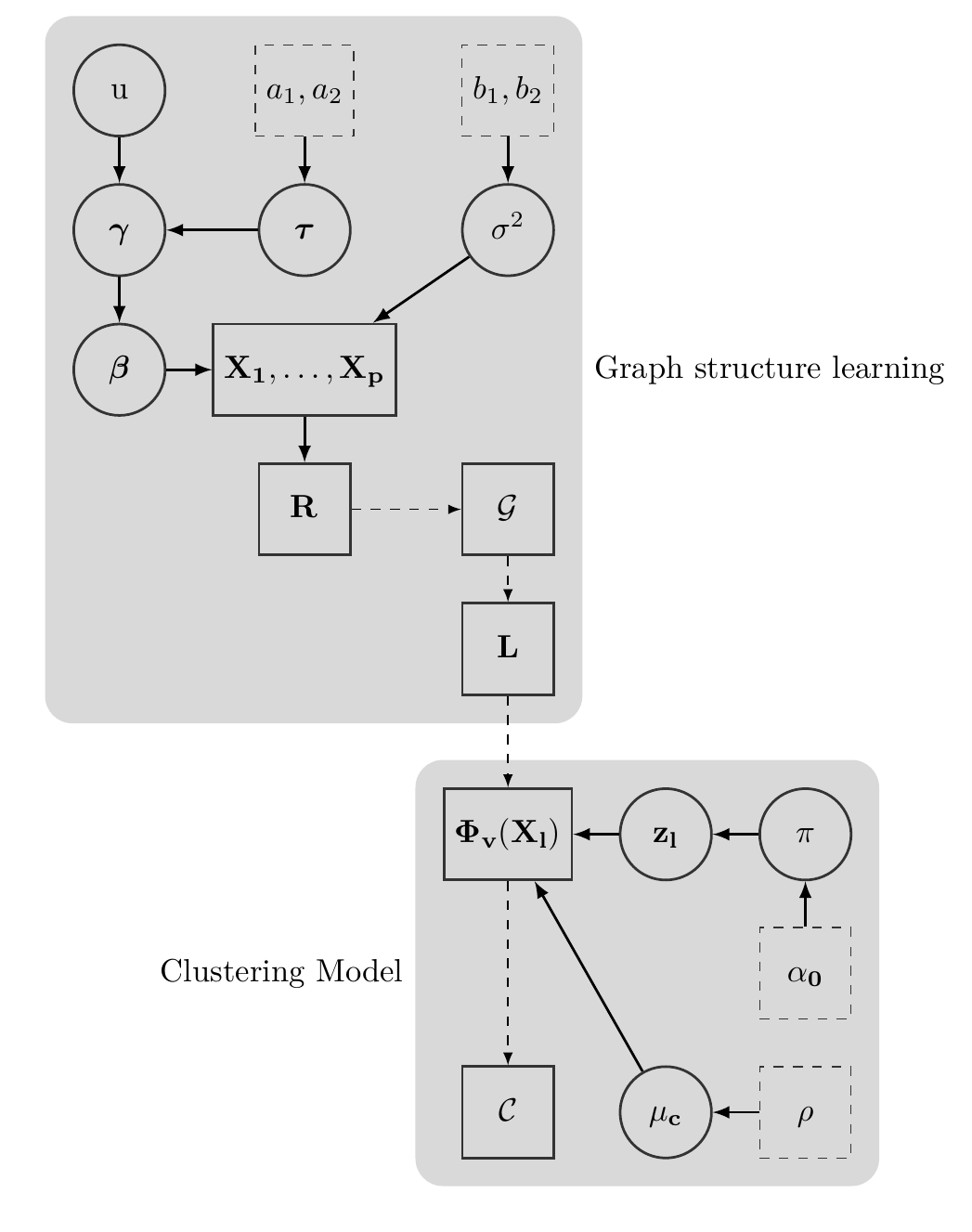}
\caption{Directed acyclic graphical representation of the Bayesian nonparametric graph clustering model. Circles represent stochastic model parameters, solid rectangles represent data and deterministic variables, dashed rectangles represent model constants, and solid and dashed arrows respectively represent stochastic and deterministic relationships.}
\label{fig:dagmodel}
\end{figure}

\subsubsection{Hard clustering via DP-means}

While MCMC approaches have benefits, one of the major drawbacks is their scalability for large $p$. Also, since our primary inference is on cluster indices for the $p$ variables, per se, rather than parameter estimates, we propose to utilize a fast computing method that evaluates the cluster indices corresponding to the above model formulation. \cite{kulis2012revisiting} considered an asymptotic version of the DPMM and developed an algorithm that gives the cluster assignments of the data points along with the cluster centers, without using the MCMC approach. Their method is closely related to the standard k-means algorithm, but with an additional penalty parameter on the number of clusters in the objective function. Their method, called \emph{DP-means}, is based on the Gibbs sampler algorithm for the DPMM, and is derived using small variance asymptotics. 

Under the previously mentioned assumptions that the covariances $\Sigma_c$ are fixed to be $\sigma I$, and the prior distribution $G_0$ over the means is taken to be $\mathrm{N}(\bm{0}, \rho I)$, for some $\rho >0$, the Gibbs probabilities can be computed explicitly. To implement DP-means, the base parameter $\alpha_0$ is further assumed to be functionally dependent on $\sigma$ and $\rho$ as $\alpha_0 = (1 + \rho / \sigma)^{K_n/2} \exp(-\lambda/2\sigma)$, for some $\lambda > 0$. The probability of assigning the data point $l$ to cluster $c$ is then proportional to 
\begin{equation}
n_{c,-l} \exp\left(-\frac{1}{2\sigma}\| \by_l - \mu_c\|_2^2 \right),
\end{equation}
where $n_{c,-l}$ denotes the number of data points already in cluster $c$, as introduced previously. Also, the probability for a new cluster is proportional to
\begin{equation}
\exp\left\{-\frac{1}{2\sigma}\left( \lambda + \frac{\sigma}{\rho + \sigma} \| \by_l \|_2^2 \right) \right\}.
\end{equation}
Then, as $\sigma \rightarrow 0$ for a fixed $\rho$, the probability of data point $i$ being assigned to cluster $c$ goes to $1$ when $\mu_c$ is closest to $\by_l$. One major advantage of using this procedure is scalability, especially in higher dimensions, where MCMC mixing might be slow. The DP-means algorithm performs similarly to MCMC-based methods, but saves time. The DP-means algorithm minimizes the objective function,
\begin{equation}
\sum_{c=1}^k \sum_{\by \in l_c} \|\by - \theta_c\|^2 + \lambda k,
\end{equation}
with $\theta_c = \frac{1}{l_c} \sum_{\by \in l_c} \by $, $l_1, \ldots, l_k$ being the resulting $k$ clusters. The penalty parameter $\lambda$ controls the number of clusters.

\section{Consistency of the graph Laplacian}
\label{Sec:consistency}

One of the main contributions of this paper is to establish theoretical guarantees of our modeling endeavor. A critical assumption in this regard is that the eigenspace of the estimated graph Laplacian should well approximate the true eigenspace corresponding to the true graph Laplacian obtained from the true partial correlation matrix. We establish this using two steps. In the first step, we show that the operator norm of the difference between these two Laplacians is bounded and it is subsequently argued that the eigenvectors corresponding to the smallest eigenvalues chosen for embedding are close to their true counterparts. This ensures that the  connected components in the graph are identified accurately.

We first define the notations to be used for the main theoretical results in this paper.

By $r_n = O(\delta_n)$ (respectively, $o(\delta_n)$), we  mean that $r_n/\delta_n$ is bounded (respectively, $r_n/\delta_n \rightarrow 0$ as $n \to \infty$). For a random sequence $X_n$, $X_n = O_P(\delta_n)$ (respectively, $X_n = o_P(\delta_n)$) means that $\mathrm{Pr}(|X_n| \leq M\delta_n) \rightarrow 1$ for some constant $M$ (respectively, $\mathrm{Pr}(|X_n| < \epsilon\delta_n) \rightarrow 1$ for all $\epsilon > 0$).  

For a vector $\bx \in \RR^p$, we define the following norms:
\begin{equation*}
\|\bx\|_r = \left(\sum_{j=1}^p|x_{j}|^r\right)^{1/r},  \; 1\le r<\infty, \quad \|\bx\|_\infty = \mathop{\max }_{j}|x_{j}|. 
%\nonumber
\end{equation*}

If $A$ is a symmetric $p\times p$ matrix, let $\eig_1(A) \le \cdots \le \eig_p(A)$ stand for its ordered eigenvalues. Viewing $A$ as a vector in $\RR^{p^2}$ and an operator from $(\RR^p,\|\cdot\|_r)$ to $(\RR^p, \|\cdot\|_s)$, where $1\le r,s \le \infty$, we have the following norms on $p\times p$ matrices:
\begin{eqnarray*}
&\|A\|_r =\left(\sum_{i=1}^p |a_{ij}|^r\right)^{1/r}, \; 1\le r<\infty, \quad \|A\|_\infty = \mathop{\max }_{i,j}|a_{ij}|,  &\nonumber \\
&\|A\|_{(r,s)} = \mbox{sup}\{\|A\bx\|_s:\|\bx\|_r = 1\}. &\nonumber
\end{eqnarray*}
The norms $\|\cdot\|_r$ and $\|\cdot\|_{(r,r)}$ are referred to as the $L_r$-norm and the $L_r$-operator norm, respectively.
Thus, we obtain the Frobenius norm as the $L_2$-norm given by $\|A\|_2 = \sqrt{\mathrm{tr}(A^TA)}$. Also, 
\begin{eqnarray*}
&\|A\|_{(1,1)} = \mathop{\max }_{j} \sum_i |a_{ij}|, \quad \|A\|_{(\infty,\infty)} = \mathop{\max }_{i} \sum_j |a_{ij}|, &\nonumber \\
&\|A\|_{(2,2)} = \{\max(\eig_i(A^TA):1\le i\le p)\}^{1/2}, \nonumber
\end{eqnarray*}
and for symmetric matrices, $\|A\|_{(2,2)}= \max\{|\eig_i(A)|:1\le i\le p\}$, and $\|A\|_{(1,1)} = \|A\|_{(\infty,\infty)}$.

\subsection{Determining bounds for the estimated graph Laplacian}
We first present the results in which we show that the operator norm of the difference between the estimated graph Laplacian and the true graph Laplacian is bounded, which implies the closeness of the respective eigenvalues so that the dimensions of the embedded data points are identical for the estimated and true Laplacian embeddings. 
\begin{theorem}
Consider the normalized graph Laplacian $L = I -  D^{-1/2} W D^{-1/2}$ based on the estimated weighted adjacency matrix $W$ and the true graph Laplacian $\mcL = I - \mcD^{-1/2} \mcW \mcD^{-1/2}$, where $\mathcal{D}$ and $\mathcal{W}$, respectively, are the true degree matrix and weighted adjacency matrix. Then, under the assumptions that the minimum degree is bounded below by $p^{1/2}(\log p)^{-1/2}$ and the maximum degree is bounded above by $p^\kappa, 1/2 \leq \kappa < 1$, we have
\begin{equation*}
\| L - \mathcal{L} \|_{(2,2)} \leq (\log p)^{3/4} p^{\kappa - 1/4} \| W - \mathcal{W} \|_2.
\end{equation*}
\label{theorem:Lapbound}
\end{theorem}

\begin{remark}
In the present context of clustering, we use the absolute value of the estimated partial correlation matrix as the weighted adjacency matrix $W$. The corresponding true weighted adjacency matrix $\mathcal{W}$ is the true partial correlation matrix $\mathcal{R}$. We have proposed to use a Bayesian neighborhood selection approach to estimate the partial correlation matrix. Consistency of the regression parameters would lead to the posterior consistency of the partial correlation matrix as well. Note that $\| W - \mathcal{W} \|_2 \leq \| R - \mathcal{R} \|_2$; hence, the posterior consistency of $R$ leads to the posterior consistency of the estimated graph Laplacian $L$. In fact, $\| L - \mathcal{L} \|_{(2,2)} = o_P(1)$ if $(\log p)^{3/4} p^{\kappa - 1/4}  \| R - \mathcal{R} \|_2 = o_P(1)$. 
\end{remark}

We present a proof of the above result in the Appendix. The bounds given above involve the difference between the estimated and true partial correlation matrices. Hence, the consistency of the graph Laplacian is dependent on the consistency of the partial correlation matrix. Apart from the Bayesian neighborhood selection approach in our paper, other consistent estimators of the partial correlation matrix or precision matrix would lead to the consistency of the graph Laplacian as well. This includes Bayesian estimation using Wishart priors, which are known to be consistent in the operator norm. 

Now we argue that the dimension of the embedded data points is identical to that obtained using the true Laplacian embedding. To show this, it suffices to argue that the eigenvalues of the estimated graph Laplacian and those of the true graph Laplacian are sufficiently close. By Weyl's inequality, we have
\begin{equation*}
\max | \eig_i(L) - \eig_i(\mcL) | \leq \| L - \mcL \|_{(2,2)} = o_P(1).
\end{equation*}
Thus, the eigenvalues of $L$ are close to those of $\mcL$. As long as the eigen-gap (denoted by $\delta$, say) between the first non-zero population eigenvalue and the corresponding zero eigenvalue(s) is large, the number of zero eigenvalues of $L$ (or eigenvalues within a distance $\epsilon_n$ from zero, where $\epsilon_n \rightarrow 0$ is the rate of contraction of $L$, dependent on the data) is equal to that of $\mcL$. In fact, the condition that $\delta > 2\epsilon_n$ suffices in this context for the dimensions of the embedded spaces to be identical.

\subsection{Closeness of the eigenvectors}

We now argue that the eigenvectors corresponding to these zero eigenvalues are also close to their population counterparts. We are dealing with orthonormal vectors (so that both are of equal magnitude), and hence evaluating the closeness of two such vectors reduces to determining the bounds of the principal angle between them. The theorem below presents the result on these bounds. The result involves the $d$ zero eigenvalues of the true graph Laplacian $\mathcal{L}$ and the corresponding $\epsilon_n$-small eigenvalues of the estimated graph Laplacian $L$. 

\begin{theorem}
Let $V = (v_1,\ldots,v_d) \in \RR^{p \times d}$ be the eigenvectors corresponding to the $d$ zero eigenvalues of $\mcL$ and $\hat{V} = (\hat{v}_1, \ldots, \hat{v}_d) \in \RR^{p \times d}$ be the eigenvectors corresponding to the $d$ zero (or $\epsilon_n$ close) eigenvalues of $L$. Also, denote by $\delta$ the eigen-gap between the zero eigenvalues and the first non-zero eigenvalue of $\mathcal{L}$, and assume that $\delta > 2\epsilon_n$. Then the principal angle between $V$ and $\hat{V}$ can be bounded as
\begin{equation*}
\| \sin \Theta(\hat{V}, V) \| _ {(2,2)} \leq \frac{\| L - \mathcal{L} \| _{(2,2)}}{\delta}.
\end{equation*}
\end{theorem}

The result easily follows that presented by \cite{davis1970rotation}, which has been recently improved by \cite{samworth2014}, involving assumptions on population eigenvalues only. The improvement is significant in the sense that we only need to care about the eigen-gap in the true graph Laplacian $\mathcal{L}$, not in the estimated one. We present their result in the theorem below, followed by the proof of our result above.

\begin{theorem}[Modified version of Davis-Kahan theorem \citep{samworth2014}]
Let $A, \hat{A} \in \RR^{p \times p}$ be symmetric, with eigenvalues $\lambda_1 \leq \ldots \leq \lambda_p$ and $\hat{\lambda}_1 \leq \ldots \leq \hat{\lambda}_p$, respectively. Fix $1 \leq r \leq s \leq p$ and assume that $\delta = \min(\lambda_r - \lambda_{r-1}, \lambda_{s+1} - \lambda_s) > 0$, where $\lambda_{p+1} := \infty, \lambda_0 := - \infty$. Let $d := s-r+1$, and let $V = (v_r,v_{r+1},\ldots,v_s) \in \RR^{p \times d}$ and $\hat{V} = (\hat{v}_r, \hat{v}_{r+1}, \ldots, \hat{v}_s) \in \RR^{p \times d}$ have orthonormal columns satisfying $Av_j = \lambda_j v_j$  and $\hat{A} \hat{v}_j = \hat{\lambda_j} \hat{v}_j, \, j = r, r+1, \ldots, s.$ Then,
\begin{equation*}
\| \sin \Theta(\hat{V}, V) \| _ {(2,2)} \leq \frac{\| \hat{A} - A \| _{(2,2)}}{\delta}.
\end{equation*}
\end{theorem}

In our context, we take the first $d$ zero eigenvalues of the graph Laplacian, so that $r = 1, s = d$ in the above theorem. The matrices $A$ and $\hat{A}$ are taken to be $\mcL$ and $L$, respectively, so that the sine of the principal angle between the eigenvectors is bounded above by the difference in the operator norm of the graph Laplacians. From Theorem \ref{theorem:Lapbound}, it follows that the $L_2$-operator norm of the sine of the principal angle between the eigenvectors goes to zero with high probability. 

Thus, in summary, we have proved that $L$ well approximates $\mathcal{L}$ with respect to the eigenspace, and hence in light of Proposition~\ref{Prop:Lap}, $L$ can be used as an effective tool to cluster the variables with graphical dependencies.
\section{Simulations}
\label{Sec:Simulation}

We perform simulation studies to evaluate the performance of our Bayesian nonparametric graph clustering (BNGC, henceforth)  method in scenarios with varying dimensions and to compare our method with other competing methods. Specifically, we evaluate the utility of using the partial correlation matrix for defining the weighted adjacency matrix for graphs and the effectiveness of using a Bayesian nonparametric model for cluster identification.

\noindent \underline{\bf Simulation design} We consider $p$-dimensional random variables with $p = 100, 200, 500$, and sample size $n = 100, 200$ for each $p$, with includes both $n \approx p$ and $n>p$ scenarios. The true number of clusters $K$ satisfies $1 \leq K \leq p$. We consider different values of $K$ such that $K = \floor{\frac{p}{5}},  \floor{\frac{p}{10}}, \floor{\frac{p}{20}}$. For each $K$, we randomly partition the $p$ variables in $K$ non-empty clusters, with all partitions having equal probability of occurrence. We simulate 10 such partitions for each $K$. For a given partition $\{\bX_1,\ldots,\bX_K\}$ of $\bX$, we generate data from a Gaussian mixture model as
\begin{equation}
f(\bx) = \prod_{j=1}^{K}f_j(\bx_j).
\end{equation}
Here, $f_j$ is a multivariate normal distribution of dimension $p_j$, with mean $\bm{0}$ and variance $\Sigma_j$. $\Sigma_j$ follows a Wishart distribution with degrees of freedom $p_j + 1$ and scale matrix identity, where $p_j$ denotes the size of cluster $j,\, 1\leq j \leq K$, and $\sum_{j=1}^K p_j = p$.

For each $n, p$ and $K$, we simulate 100 data sets and apply our method to identify the clusters. For construction of the adjacency matrices, we use the `spikeslab' package in \texttt{R}, applying the default parameter settings. For graph clustering using DPMM, we resort to the DP-means algorithm for faster computation, choosing the penalty parameter $\lambda$ for the number of clusters to be 1/2. We compute normalized mutual information (NMI) scores that correspond to each combination of $n,p,K$ to assess the performance of our method. NMI scores are useful for evaluating clustering performances when the true cluster labels are known. We denote the true cluster labels as $\mathcal{C}$ and the cluster labels obtained from a suitable clustering method as $\mathcal{\hat C}$. Note that both $\mathcal{C}$ and $\mathcal{\hat C}$ are partition sets of the same set of $p$ variables. As a measure of uncertainty, the entropy of a partition set $T = \{T_1,\ldots,T_s\},\, s = |T|$ is defined as 
$$ \mathrm{H}(T) = \sum_{i = 1}^{|T|} \frac{|T_i|}{p} \log \frac{|T_i|}{p}.$$ 
Also, the mutual information between $\mathcal{C}$ and $\mathcal{\hat C}$ is given by 
$$\mathrm{MI}(\mathcal{C}, \mathcal{\hat C}) = \sum_{i=1}^{|\mathcal{C}|} \sum_{j=1}^{|\mathcal{\hat C}|} \frac{|C_i \cap \hat{C}_j |} {p} \log \frac{ p |C_i \cap \hat{C}_j |}{| C_i | |\hat{C}_j |},$$ where $C_i$s and $\hat{C}_j$s are elements of the sets $\mathcal{C}$ and $\mathcal{\hat C}$, respectively.
The normalized information score between $\mathcal{C}$ and $\mathcal{\hat C}$ is then defined as 
$$\mathrm{NMI} ( \mathcal{C}, \mathcal{\hat C} ) = \frac{\mathrm{MI}(\mathcal{C}, \mathcal{\hat C})} {\sqrt {\mathrm{H}(\mathcal{C}) \mathrm{H} (\mathcal{\hat C}) }}.$$
NMI scores are bounded between 0 and 1, with zero implying complete independence of cluster labelings, and values close to one indicating significant agreement between the clusters. 

In addition to NMI scores, we compute the between-cluster edge densities to assess the overall clustering performance of our methods in finding tightly connected graph components. The between-cluster edge density for a partition set $\mathcal{\hat C}$ is given by
$$ \mathrm{edge\,density}(\mathcal{\hat C}) = \mathop{\sum_{X_i \in \hat{C}_{k_1}, \, X_j \in \hat{C}_{k_2}}}_{k_1 \neq k_2} w_{ij},$$
where $w_{ij}$ is the edge weight between $X_i$ and $X_j$ for the corresponding graphical model. Between-cluster edge densities are bounded below by zero, with zero implying that the partitioning of the graph is such that there are no edges shared by the vertices in different clusters. In our simulation design, the true graphs are such that between-cluster edge densities are perfectly zero; hence, the above measure is a good tool to use in evaluating the clustering performance of a method. Deviation of the above measure from zero indicates non-agreement with the true cluster labels.

\noindent \underline{\bf Comparison with competing methods}. To benchmark our results, we use three competing methods: k-means, graphical kernel-based (graph-kernel) and hierarchical clustering with average-linkage (ALC) methods. The first two are standard spectral clustering algorithms. The k-means method uses a weighted adjacency matrix that is identical to ours. We also compare our method with spectral algorithms applied to the adjacency matrices obtained from the data using kernel-based similarity measures. We use the \texttt{kernlab} package in \texttt{R} for fitting the latter method. Note that both spectral clustering methods require the specification of the number of clusters beforehand. The comparison to the k-means method shall assess the performance of using a nonparametric Bayesian model to arrive at the final cluster labels, as they both utilize the identical adjacency matrix. The comparison to the graph-kernel method shall assess the performance of the adjacency matrix constructed with the estimated partial correlation matrix from the Bayesian neighborhood selection approach, with regular adjacency matrices used in most of the applications. For the ALC, we use the empirical correlation matrix and provide the true number of clusters while performing hierarchical clustering.  Note that for all the competing methods, we need to supply the number of clusters, and we provide the \underline{true number of clusters} while assessing the performance, thus giving the methods a fair advantage. In comparison, our BNGC method learns the number of clusters from the data set, itself.

\noindent \underline{\bf Results.} The complete simulation results using the different metrics and scenarios are presented in Table 1. We observe that for all the methods, the NMI scores increase with an increase in sample size for fixed $p$ and $K$. However, the NMI scores corresponding to BNGC and k-means methods indicate that these two methods have significantly better performance than the graph-kernel and ALC methods across all $n, p, K$ scenarios. In fact, even for $n=200$, we observe NMI scores above $0.95$ for all $p$ and $K$ when using the BNGC and k-means methods, but the graph-kernel and ALC methods fail to achieve such a mark. Hence, the results favor the partial correlation matrix as the adjacency matrix for graph clustering. The performances of the methods with respect to the NMI scores are shown in Figure \ref{fig:nmiplots}.

The superior performances of the BNGC and k-means methods are also reflected in the edge densities between the clusters. For both methods that use the partial correlation matrix as the adjacency matrix, the between-cluster edge densities decrease with increasing sample sizes, but no such trend is seen for the remaining two methods. Also, although the edge densities for the BNGC and k-means methods shrink to zero, implying a significant absence of edges between clusters, the corresponding edge densities for the graph-kernel and ALC methods are much higher than zero, implying a poor clustering performance. This result further strengthens the use of a partial correlation matrix as an effective tool for graph clustering in such cases. We conjecture that the dependence structure in a graph is more efficiently captured by the partial correlation matrix in comparison with other adjacency matrices, such as the kernel-based adjacency matrix.

Although BNGC and k-means methods perform similarly, the performance of the BNGC method is slightly superior.  Also, the measure of the between-cluster edge densities goes to zero faster under our method than under the k-means method. This implies that cluster misspecification is lower under the BNGC method. In addition, the BNGC method performs better than the other methods when choosing the number of clusters. 

\begin{figure}
\centering
\includegraphics[scale=0.9]{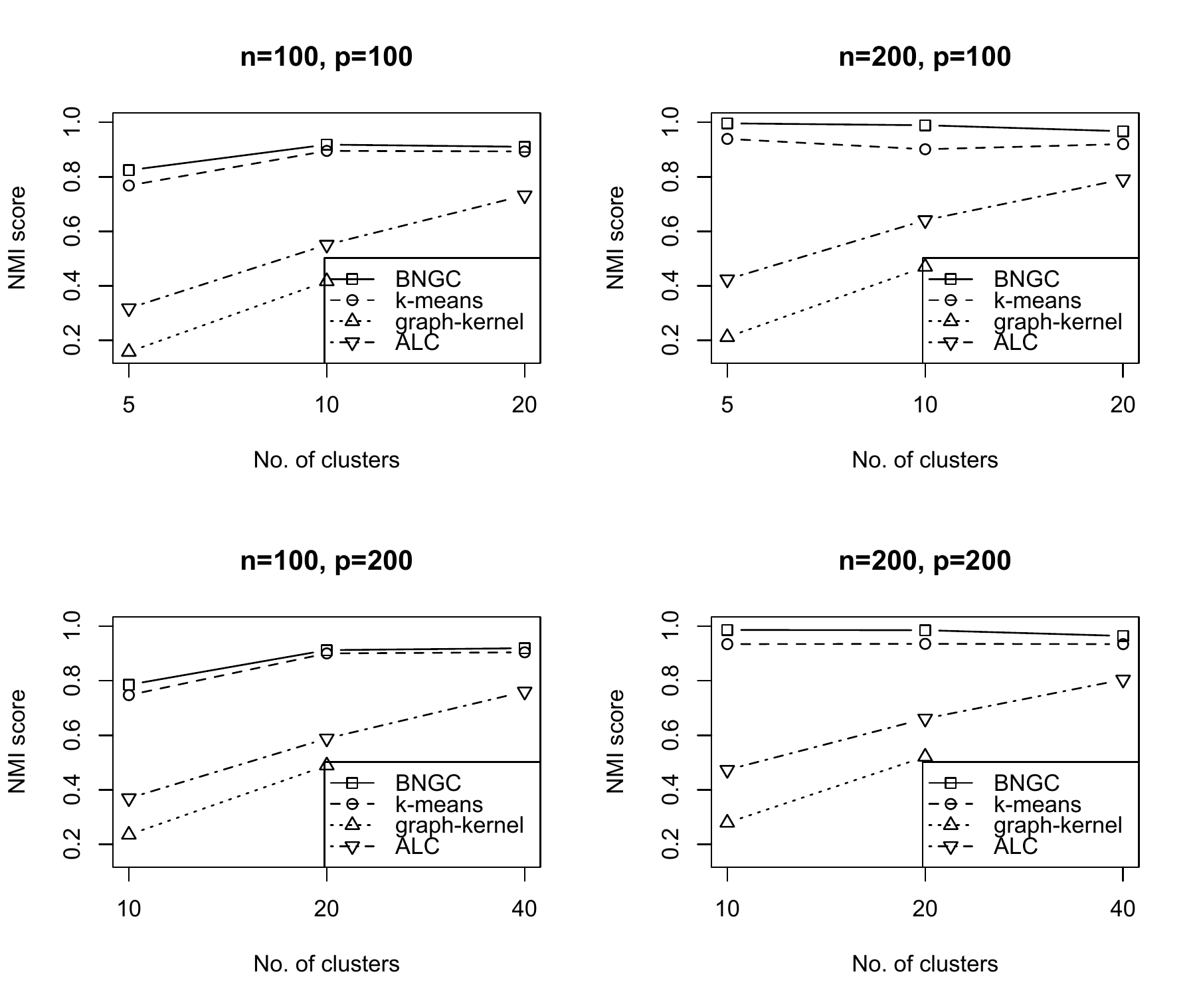}
\caption{Normalized mutual information (NMI) scores for different values of $n$ and $p$ that correspond to the various clustering methods.}
\label{fig:nmiplots}
\end{figure}

\begin{sidewaystable}
\footnotesize
\begin{center}
\label{Table:Table1}
{
\caption{Clustering performances of different methods for simulated data sets. Results have been averaged over $10$ random partitions, each having $100$ data sets for each choice of the number of clusters $K$ corresponding to different values of $n$ and $p$. Figures in parentheses denote the standard errors.}
\begin{tabular}{cccccccccccc}
\hline
&&& \multicolumn{3}{c}{NMI Score} & & \multicolumn{3}{c}{Edge density} \\
\cline{4-7}  \cline{9-12} 
n & p & K & k-means & BNGC & graph-kernel & ALC & & k-means  & BNGC & graph-kernel & ALC\\
\hline
\multirow{3}{*}{100} & \multirow{3}{*}{100} & 5 & 0.768 (0.158) &0.824 (0.170) &0.158 (0.049)  & 0.318 (0.074) & & 0.060 (0.061) &0.036 (0.040) &1.432 (0.151)  & 0.751 (0.097) \\
& & 10 & 0.895 (0.050) &0.918 (0.048) &0.417 (0.034) & 0.551 (0.049) & & 0.011 (0.010) &0.003 (0.003) &0.329 (0.031) & 0.168 (0.022)\\
& & 20 & 0.893 (0.030) &0.910 (0.039) & $\dagger$ \footnote{$\dagger$ implies that the method failed to produce any output due to memory problems} & 0.732 (0.027) & & 0.006 (0.004) &0.000 (0.000) & $\dagger$ & 0.032 (0.085)\\
&&&&&&&\\
\multirow{3}{*}{100} & \multirow{3}{*}{200} & 10 & 0.747 (0.098) &0.785 (0.118) &0.235 (0.037) & 0.369 (0.036) & & 0.073 (0.051) &0.041 (0.041) &0.619 (0.048) & 0.353 (0.026)\\
& & 20 &0.900 (0.033) &0.912 (0.043) &0.489 (0.026) & 0.588 (0.024) && 0.007 (0.004) &0.002 (0.002) &0.145 (0.010) & 0.079 (0.007)\\
& & 40 & 0.904 (0.020) &0.919 (0.022) & $\dagger$ & 0.760 (0.015) && 0.003 (0.001) &0.000 (0.000) & $\dagger$ & 0.015 (0.002)\\
&&&&&&&\\
\multirow{3}{*}{100} & \multirow{3}{*}{500} & 25 & 0.676 (0.092) &0.650 (0.123) &0.240 (0.031) & 0.433 (0.016) && 0.132 (0.108) &0.055 (0.045) &0.508 (0.032) & 0.125 (0.005)\\
& & 50 &0.879 (0.016) &0.882 (0.025) &0.520 (0.021) & 0.628 (0.010) && 0.005 (0.002) &0.002 (0.001) &0.295 (0.013) & 0.029 (0.001)\\
& & 100 & 0.882 (0.013) &0.887 (0.018) & $\dagger$ & 0.770 (0.007) && 0.002 (0.000) &0.000 (0.000) & $\dagger$ & 0.006 (0.000) \\
&&&&&&&\\
\multirow{3}{*}{200} & \multirow{3}{*}{100} & 5 & 0.939 (0.066) &0.996 (0.010) &0.212 (0.060) & 0.424 (0.096) && 0.060 (0.082) &0.005 (0.004) &2.280 (0.198)  & 1.297 (0.207)\\
& & 10 &  0.901 (0.091) &0.989 (0.015) &0.470 (0.051) & 0.641(0.055) & & 0.035 (0.036) &0.001 (0.001) &0.493 (0.051) & 0.267 (0.043)\\
& & 20 &  0.920 (0.039) &0.967 (0.019) & $\dagger$ & 0.791 (0.029) & & 0.010 (0.008) &0.000 (0.000) & $\dagger$ & 0.050 (0.008)\\
&&&&&&&\\
\multirow{3}{*}{200} & \multirow{3}{*}{200} & 10 & 0.934 (0.042) &0.986 (0.023) &0.279 (0.043)  & 0.473 (0.046) && 0.030 (0.025) &0.004 (0.002) &1.035 (0.063) & 0.622 (0.051)\\
& & 20 &0.935 (0.038) &0.985 (0.010) &0.521 (0.027) & 0.661 (0.027) && 0.014 (0.009) &0.001 (0.001) &0.229 (0.013) & 0.134 (0.013)\\
& & 40 & 0.934 (0.016) &0.964 (0.015) & $\dagger$ & 0.804 (0.014) & &  0.005 (0.003) &0.000 (0.000) & $\dagger$ & 0.025 (0.002)\\
&&&&&&&\\
\multirow{3}{*}{200} & \multirow{3}{*}{500} & 25 & 0.929 (0.023) &0.988 (0.011) &0.232 (0.031) & 0.526 (0.021) && 0.012 (0.005) &0.003 (0.001) &0.519 (0.042) & 0.234 (0.009)\\
& & 50 &0.953 (0.011) &0.981 (0.007) &0.427 (0.036) & 0.691 (0.013) && 0.004 (0.001) &0.000 (0.000) &0.121 (0.016) & 0.051 (0.002)\\
& & 100 & 0.935 (0.011) &0.954 (0.009) & $\dagger$ & 0.809 (0.008)  && 0.002 (0.001) &0.000 (0.000) & $\dagger$ & 0.010 (0.001) \\
\hline
\end{tabular}
}
\end{center}
\end{sidewaystable}

\section{Proteomic signaling networks in cancer}
\label{Sec:Real data}

\subsection{Scientific problem and data description}

Proteins are the ultimate effector molecule of cellular functions. Proteomics, in general, can be defined as a large-scale high-throughput study of proteins from a variety of biological samples in order to investigate their ontology, classification, expression levels, and properties. A primary functional proteomic technology is reverse-phase protein array (RPPA), which allows for quantitative, high-throughput, time- and cost-efficient analysis of proteins using small amounts of biological material \citep{paweletz01,tibes2006}. RPPA allows for the simultaneous assessment of multiple protein markers in multiple tumor samples in a streamlined and reproducible manner \citep{sheehan2005,spurrier2008}. This technology has been extensively validated for both cell line and patient samples \citep{tibes2006,hennessy2010,nishizuka2003}, and its applications range from building reproducible prognostic models \citep{yang2013} to generating experimentally verified mechanistic insights \citep{liang7} and biomarker discovery, especially in cancer \citep{hennessy2010}. For a detailed introduction, data pre-processing and normalization of RPPA data, see \cite{baladandayuthapani2014bayesian}.

It is well established  that oncogenic proteomic changes occur in a coordinated manner across multiple signaling networks and  pathways, reflecting various pathobiological processes \citep{zhang09,halaban10}.  Numerous  inhibitors of such pathways have been used in clinical trials, frequently demonstrating dramatic clinical activity. Inhibitors that target protein signaling pathways have been approved by the U.S. Food and Drug Administration for a variety of cancer types, including  leukemia, breast cancer, colon cancer, renal cell carcinoma, and gastrointestinal cancer \citep{davies06}. Thus, developing an accurate understanding of the composition (i.e., clustering) and topology (i.e., graph) of these protein signaling networks across multiple cancer  types can provide deeper biological insights about proteomic activity.

The proteomic data set we consider here was generated by RPPA  analysis of 3467 patient tumor samples across 11 cancer types, and was  obtained from The Cancer Genome Atlas (TCGA, http://cancergenome.nih.gov). The tumor samples include 747 breast cancer (BRCA), 334 colon adenocarcinoma (COAD), 130 renal adenocarcinoma (READ), 454 renal clear cell carcinoma (KIRC), 412 high-grade serous ovarian cystadenocarcinoma (OVCA), 404 uterine corpus endometrial carcinoma (UCEC), 237 lung adenocarcinoma (LUAD), 212 head and neck squamous cell carcinoma (HNSC), 195 lung squamous cell carcinoma (LUSC), 127 bladder urothelial carcinoma (BLCA), and 215 glioblastoma multiforme (GBM) samples. From each tumor sample, we have measurements of 181 different proteins that cover major functional and signaling pathways in cancer such as P13K, MAPK and mTOR. Based on the availability of protein data across a large number of tumor samples, the scientific objectives on which we focus in this paper are: (a) evaluate how the protein network topology change for each cancer type; and (b) use this information to evaluate how proteomic clusterings change within and across cancer types. These investigations will provide insights into clusters that are conserved across multiple tumors as well differential networks/clusters that are tumor-specific.

\subsection{Results}

\subsubsection{Cancer-specific clustering}
For each cancer type, we apply our nonparametric graphical clustering model, as explained in the previous sections, to obtain cancer-specific networks and corresponding clusters.  The number of clusters identified by our method for each cancer type is shown in Figure \ref{fig:numcluster}. The results show considerable cluster heterogeneity between cancer-types, with bladder cancer (BLCA) having the largest number of clusters and squamous cell lung cancer (LUSC) the lowest. Those results are consistent with previously published findings \citep{akbani2014pan}, where bladder cancer was shown to be the most heterogeneous disease among the 11 different tumor types studied, in terms of proteomic activity.

The corresponding networks of proteins for BRCA, LUSC, READ and GBM are presented in Figures \ref{fig:Ng1}, \ref{fig:Ng2}, \ref{fig:Ng3} and \ref{fig:Ng4}, respectively; the rest are provided in the Supplementary Materials (Section S4). For better visualization, we highlight clusters that include at least 4 proteins. As is evident from the figures, the within-cluster edges are considerably higher compared to the between-cluster edges, which establishes that our method performs reasonably well in picking relevant protein clusters.

\begin{figure}
\includegraphics[scale=0.6]{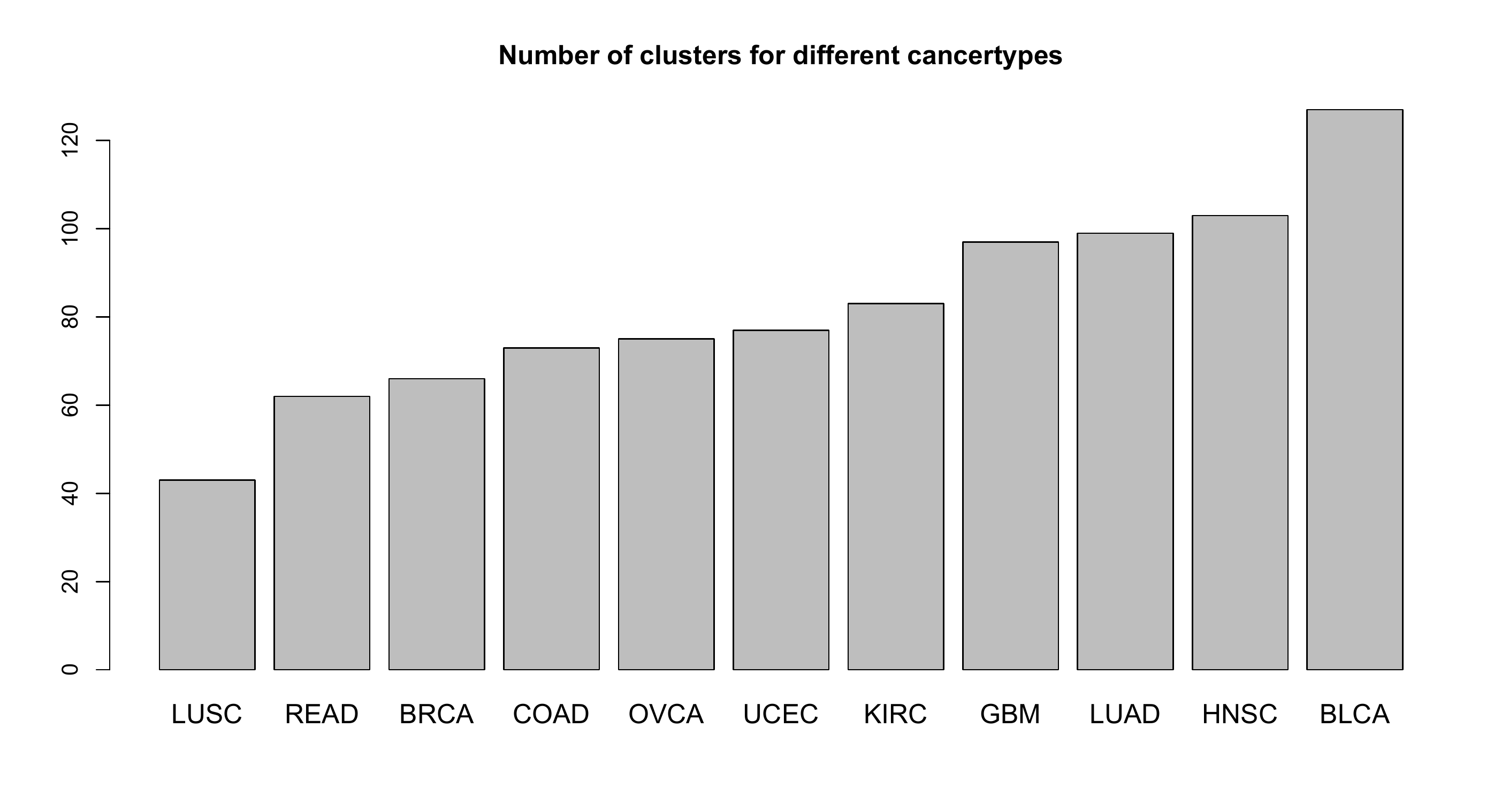}
\caption{Number of clusters identified by the BNGC method for different cancer types.}
\label{fig:numcluster}
\end{figure}

\subsubsection{Validation using known signaling pathways}

To further investigate the proteomic clusters identified by our method, we used known biological knowledge -- proteins that are mapped to existing signaling pathways. We used information from 12 signaling pathways: apoptosis, breast reactive, cell cycle, core reactive, DNA damage response, EMT, hormone receptor, hormone signaling, PI3K/AKT, RAS/MAPK, RTK, and TSC/mTOR. We then defined a {\it pathway enrichment probability} matrix $\mathcal{P}^{(t)}_E$ that corresponds to cancer type $t$ as a metric to evaluate which pathways are grouped together in different cancers. This matrix is of dimension $K_t \times 12$, where $K_t$ is the number of major clusters in cancer type $t$. For each cancer type, we follow a Bayesian hypothesis testing procedure to check whether the proportion of proteins from pathway $j$ in cluster $k$ is significantly higher than the proportion of proteins from pathway $j$ outside cluster $k$. This is a simple test to determine whether the binomial proportion is greater than $1/2$, and can be carried out using a beta-binomial model from a Bayesian perspective. Denoting the number of proteins from pathway $j$ in cluster $k$ as $y_{kj}$ and the total number of proteins in pathway $j$ as $N_j$, the test model is given by $$y_{kj}\mid \theta \sim \mathrm{Bin}(N_j , \theta),\,\theta \sim \mathrm{Beta}(1,1).$$
The posterior distribution of $\theta \mid N_j, y_{kj}$ is $\mathrm{Beta}(y_{kj} + 1, N_j - y_{kj} + 1)$. Thus, the $(k,j)^\mathrm{th}$ element of $\mathcal{P}^{(t)}_E$ for a particular cancer type $t$  is given by 
$$ \mathcal{P}^{(t)}_{E,kj} = \mathrm{Pr}( \theta > 0.5 \mid N_j, y_{kj}),$$
the posterior probability that $\theta$ exceeds $1/2$, which can be easily computed using Monte Carlo methods.

The corresponding pathway enrichment probabilities for the different cluster-cancer combinations are shown as a heatmap in  Figure \ref{fig:proenrich}. Major pathways that are enriched ($\mathcal{P}_E > 0.5$) are also shown in Table~\ref{Table:enrich}.
We find major enrichment in three pathways: hormone receptor (6 cancers), core reactive (6 cancers), and RTK (3 cancers).

\begin{table}[h]
\begin{center}
    \begin{tabular}{ cl p{5cm} l}
    \hline
    Cancer type & Enriched pathways \\
    \hline
    LUSC & Cell cycle, Core Reactive, Hormone receptor, RAS/MAPK, RTK\\
    LUAD & Core reactive, DNA damage response, RTK \\
    READ & Core reactive, Hormone receptor, TSC/mTOR \\
    COAD & Breast reactive, Core Reactive, Hormone receptor \\
    BRCA & Breast reactive, Hormone receptor, RTK \\
    OVCA & Core Reactive, RAS/MAPK \\
    UCEC & Core Reactive, Hormone receptor \\ 
    KIRC & Cell cycle, TSC/mTOR \\
    GBM & Hormone receptor \\
    BLCA & Hormone Signalling \\
    \hline
    \end{tabular}
    \caption{Table showing pathways for different cancer types with pathway enrichment probability exceeding 0.5.}
     \label{Table:enrich}
\end{center}
\end{table}

\subsection{Biological interpretation of results}

A comparison of Figure~\ref{fig:numcluster} with Table~\ref{Table:enrich} shows a roughly inverse relationship between the number of clusters versus the number of pathways identified in a given tumor type (with a few exceptions). For example, BLCA has the largest number of clusters but only one pathway, whereas LUSC has the fewest number of clusters, but the largest number of pathways (five). Part of the explanation is that we are looking for pathways that are enriched consistently across the full cohort of samples within a tumor type. If a tumor type is very heterogeneous (such as BLCA, \cite{cancer2014comprehensive}), that decreases the chances that we will find enrichment of a pathway consistently across all the samples within that tumor type, which in turn translates to fewer enriched pathways found. However, tumor type heterogeneity will result in a greater number of distinct clusters found. GBM, for instance, was previously shown to have 4 to 5 distinct subtypes \citep{verhaak2010integrated, brennan2013somatic}, but only one enriched pathway was found by us.

Several of the pathways found in Table~\ref{Table:enrich} have known biological underpinnings. For example, it is well-known that a large subset of women's cancers, such as breast (non-triple negatives) and endometrial, have elevated hormone receptors (ER and PR) \citep{perou2000molecular, sorlie2001gene, creasman1993prognostic}, which was picked up by our method. In addition, a substantial number of breast cancers have elevated HER2 levels and up-regulation of the RTK pathway as well \citep{perou2000molecular, sorlie2001gene}. The breast reactive pathway was defined using previously described breast reactive samples \citep{cancer2012comprehensive} and successfully identified by our algorithm as a validation check. The reactive markers were also found to be high in colorectal cancers (COAD/READ) in a previous study \citep{akbani2014pan}, and we have found the same pathway to be activated in COAD/READ. Hormone signaling activity related to the expression of GATA3 is high in both normal and malignant bladder samples and it is captured by our results \citep{miyamoto2012gata}. The tumor suppressor gene BTG3 has been shown to be down-regulated in renal cancers through hyper-methylation of its promoter \citep{majid2009btg3}. BTG3 is a cell cycle inhibitor, so its down-regulation increases cell cycle activity in renal cancers, which is also illustrated by our results. Other novel results, such as the role of hormone receptors in GBM, remain to be investigated in detail.

\begin{figure}
\centering
   \begin{subfigure}[b]{0.75\textwidth}
   \includegraphics[width=1\linewidth]{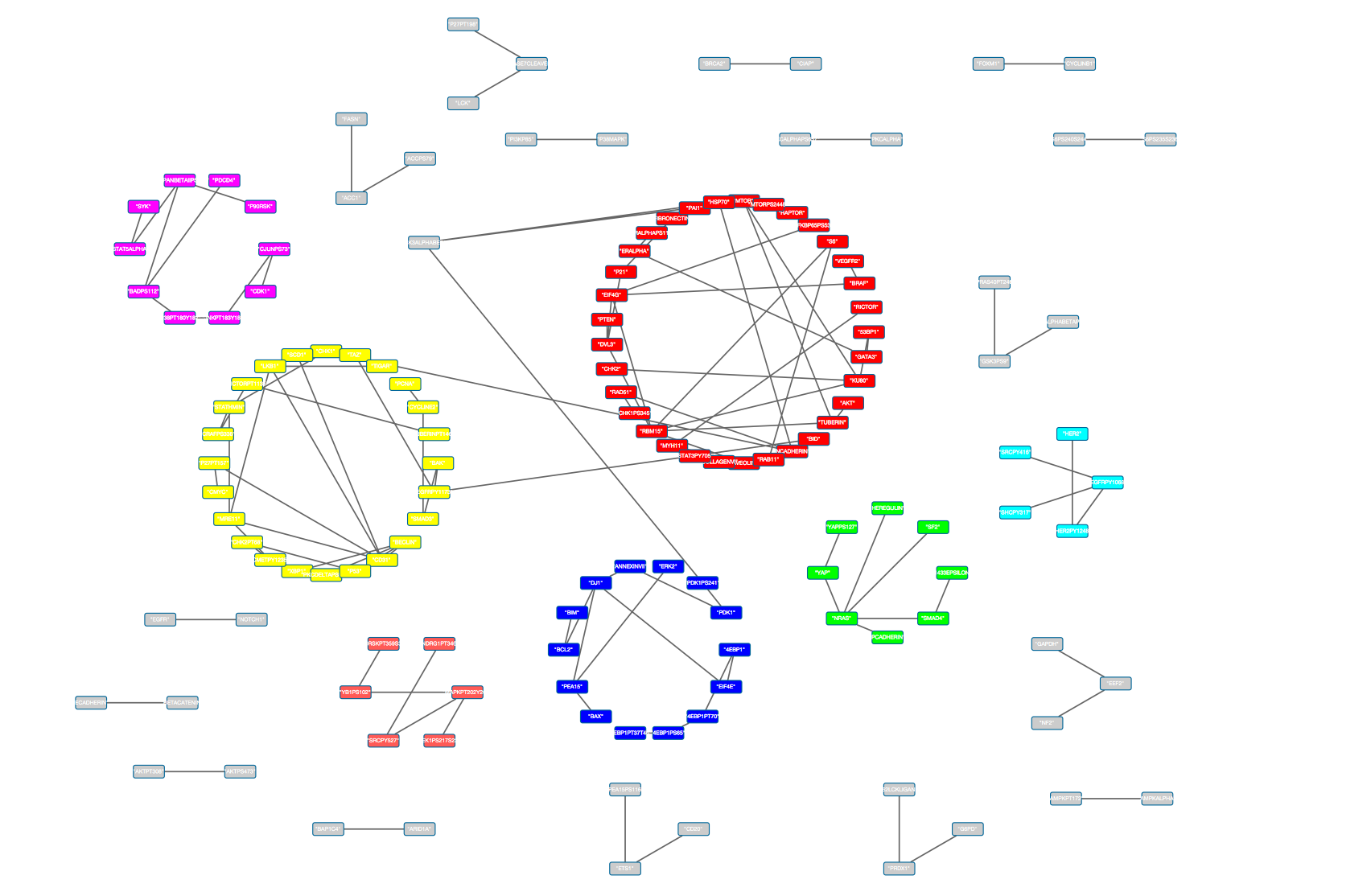}
   \caption{}
   \label{fig:Ng1} 
\end{subfigure}

\begin{subfigure}[b]{0.75\textwidth}
   \includegraphics[width=1\linewidth]{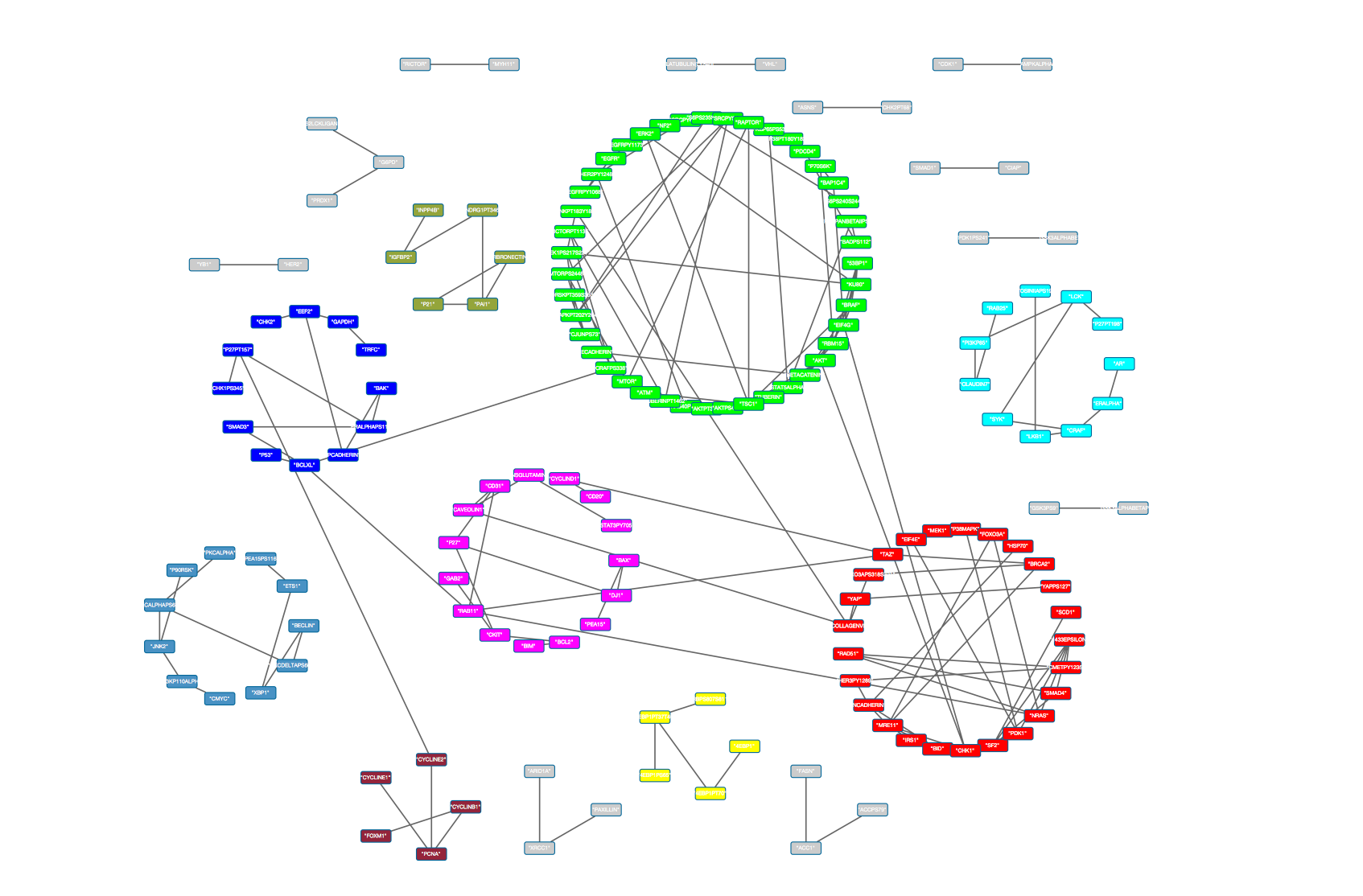}
   \caption{}
   \label{fig:Ng2}
\end{subfigure}

\caption{Protein networks for (a) breast cancer (BRCA) and (b) lung squamous cell carcinoma (LUSC). Clusters with at least 4 proteins are color-coded; those with fewer are gray.}
\end{figure}

\begin{figure}
\centering
   \begin{subfigure}[b]{0.75\textwidth}
   \includegraphics[width=1.1\linewidth]{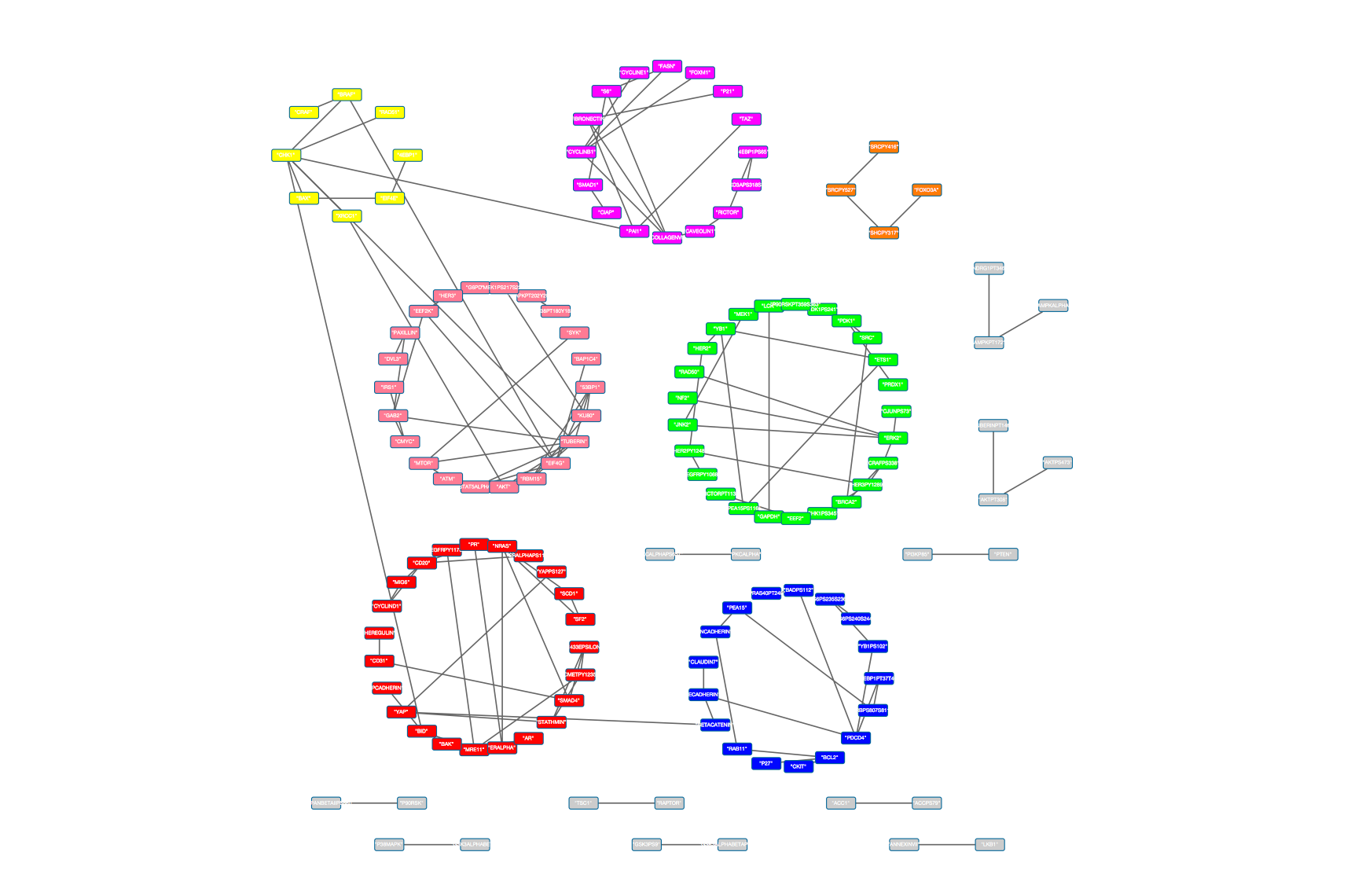}
   \caption{}
   \label{fig:Ng3} 
\end{subfigure}

\begin{subfigure}[b]{0.75\textwidth}
   \includegraphics[width=1.1\linewidth]{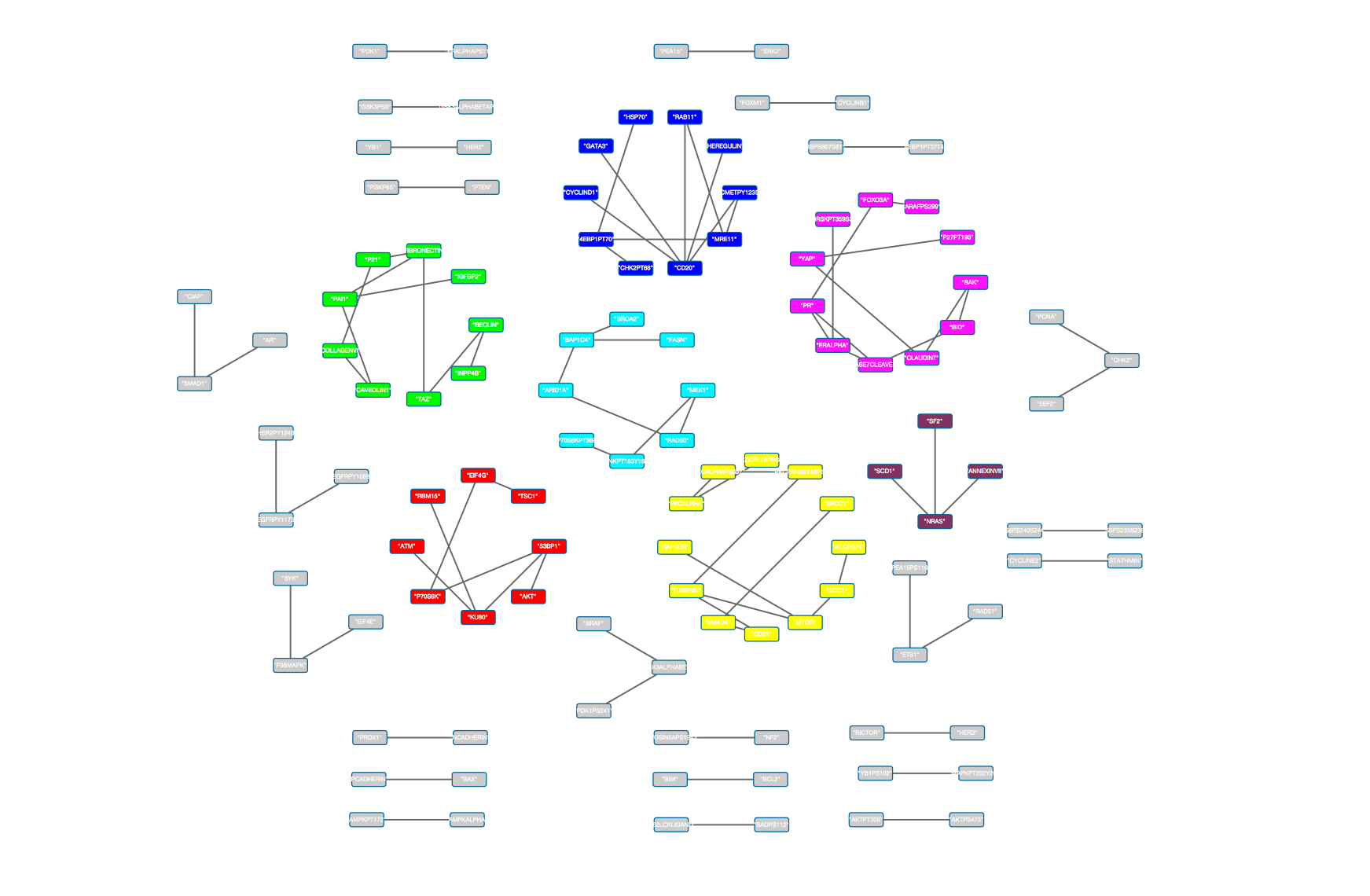}
   \caption{}
   \label{fig:Ng4}
\end{subfigure}

\caption{Protein networks for (a) renal adenocarcinoma (READ) and (b) glioblastoma multiform (GBM). Clusters with at least 4 proteins are color-coded; those with fewer are gray.}
\end{figure}

\begin{sidewaysfigure}
\includegraphics[scale=0.8]{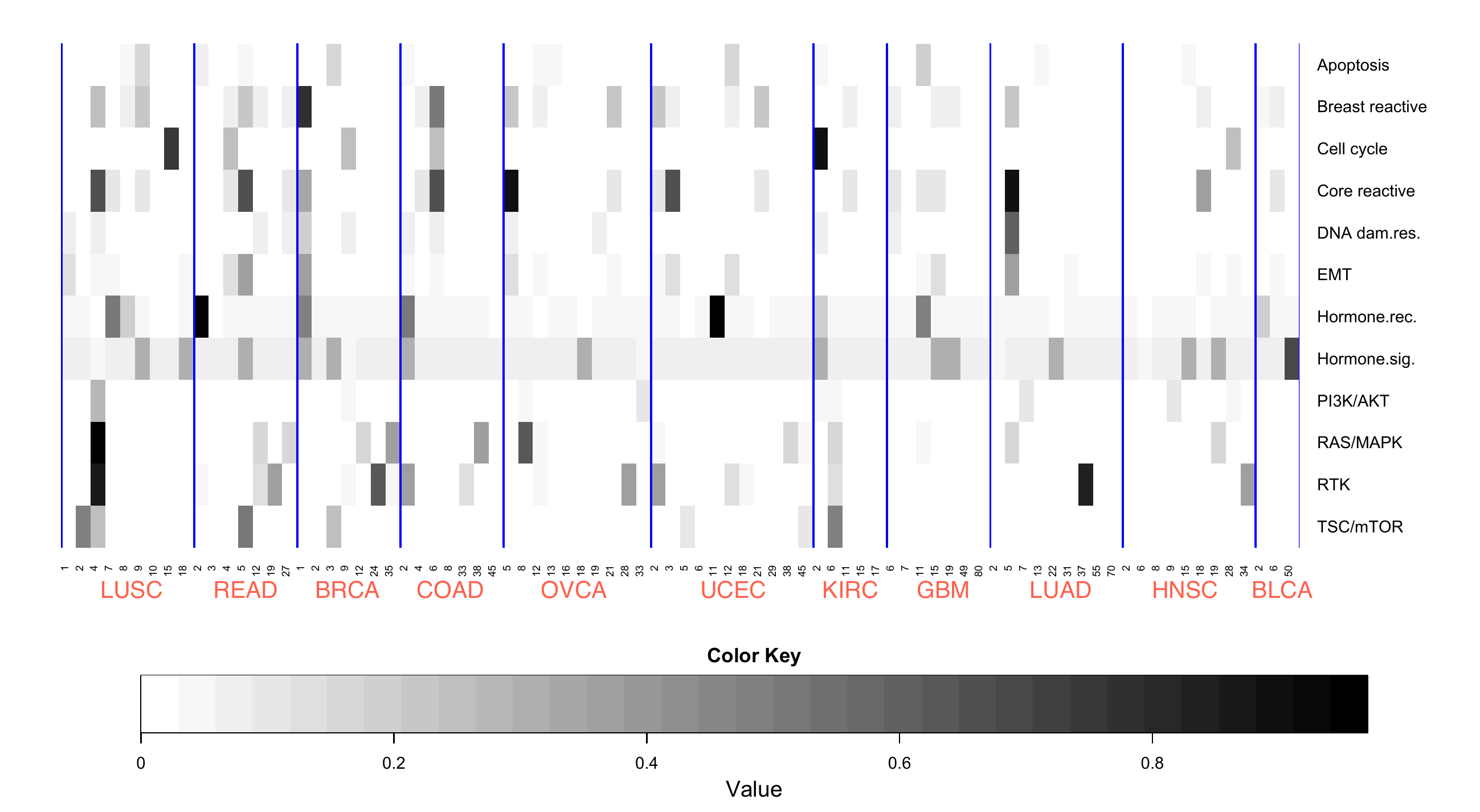}
\caption{Heatmap of pathway enrichment probabilities for different cancers. The rows represent signaling pathways and columns indicate protein clusters ($> 3$ proteins), corresponding to different cancer types and separated by blue vertical lines.}
\label{fig:proenrich}
\end{sidewaysfigure}

\section{Summary and discussion}

We develop nonparametric methods for identifying clusters of variables for moderate to high-dimensional data with graphical dependencies. Our fully probabilistic methods allow joint graphical structure learning and cluster identification utilizing the eigenspace of the underlying graph Laplacians. 
We provide rigorous theoretical justifications for choosing Laplacian embeddings as suitable projections for graph-structured data. In addition, we present fast and scalable computing techniques that are capable of handling high-dimensional data sets. Through simulations, we demonstrate superior numerical performance of our methods over standard competing methods from the literature. Our methods are motivated by and applied to a novel pan-cancer proteomic dataset where we infer tightly connected clusters of proteins across various cancer types validate them using known biological pathway-based information.

From the data analysis perspective, the proteomic data we have, consists of the same set of proteins over different subjects and cancer types. Although we analyzed each data set separately, a combined analysis across different cancer types can be performed by pooling the ``matched" data across proteins. In the Supplementary Material Section S1, we discuss a possible approach to this problem using our Bayesian graphical clustering method. The approach takes into account the graph Laplacians for each of the data sets and uses a joint model for the local clusters within the individual data sets and a single global cluster across all cancer types. An open problem in this regard is that we do not know the behavior of the Laplacians in determining the overall global clustering; we leave that for a future study. 

Another possible extension of the clustering problem discussed in this article is to work with discrete or mixed graphical models, where the variables are no longer required to be continuous, but may vary over different domains, including categorical variables. Such data sets are becoming more abundant in the literature; thus, smarter scalable methods need to be developed. In these contexts, the relations among the different variables become nonlinear; hence, we need to extend our work beyond Gaussian graphical models. Precision matrices (or partial correlation matrices) do not provide an effective way of defining the edges of the corresponding graph. We need to develop measures that can capture the nonlinear dependency among the vertices of the graph and translate the strength of association to the corresponding edges. We are currently exploring these aspects using fully nonparametric techniques for graphical models.

\section{Acknowledgements}

S.B. and V.B. were partially supported by NIH grant R01 CA160736. R.A. was supported in part by U.S. National Cancer Institute (NCI; MD Anderson TCGA Genome Data Analysis Center) grant number CA143883, the Cancer Prevention Research Institute of Texas (CPRIT) grant number RP130397, the Mary K. Chapman Foundation, the Michael \& Susan Dell Foundation (honoring Lorraine Dell). All authors were also partially supported by  MD Anderson Cancer Center Support Grant P30 CA016672 (the Bioinformatics Shared Resource).

\appendix

\section{Proofs and additional results}

We present the proof of the consistency of the graph Laplacian. The proof uses standard matrix ordering inequalities, which we present in the following lemma.
\begin{lemma}
For symmetric matrices $A$ and $B$ of order $p$, we have,
\begin{enumerate}
\item $\|A\|_{\infty} \leq \|A\|_{(2,2)} \leq \|A\|_2 \leq p\|A\|_{\infty}$,
\item $\|AB\|_{(2,2)} \leq \|A\|_{(2,2)}\|B\|_2,\, \|AB\|_{(2,2)} \leq \|A\|_2\|B\|_{(2,2)}$.
\end{enumerate}
\end{lemma}

\begin{proof}[Proof of Theorem~\ref{theorem:Lapbound}]

We show the consistency of the graph Laplacian in the operator norm. We work with the normalized Laplacian $L_{\mathrm{sym}}$, and denote it by $L$ for simplicity. Recall that $L = I - D^{-1/2}WD^{-1/2}$ is the graph Laplacian obtained from the data with weighted adjacency matrix $W = (\!(w_{lj})\!)$, and $\mcL = I - \mcD^{-1/2}\mcW\mcD^{-1/2}$ is the true Laplacian corresponding to the true adjacency matrix $\mcW = (\!(w_{lj}^0)\!)$. The $L_2$-operator norm of the difference between the matrices $L$ and $\mcL$ is given by
\begin{eqnarray}
\lefteqn{ \| D^{-1/2}WD^{-1/2} - \mcD^{-1/2}\mcW\mcD^{-1/2}\|_{(2,2)} }\nonumber \\ 
&=& \| D^{-1/2}(W - \mcW) D^{-1/2} + D^{-1/2} \mcW (D^{-1/2} - \mcD^{-1/2})  + (D^{-1/2} - \mcD^{-1/2}) \mcW \mcD^{-1/2} \|_{(2,2)} \nonumber \\
&\leq& \| D^{-1/2}(W - \mcW) D^{-1/2}\|_{(2,2)} + \|  D^{-1/2} \mcW (D^{-1/2} - \mcD^{-1/2})\|_{(2,2)}  \nonumber \\
&& + \| (D^{-1/2} - \mcD^{-1/2}) \mcW \mcD^{-1/2} \|_{(2,2)} \nonumber \\
&=& T_1 + T_2 + T_3.
\end{eqnarray}

Now,
\begin{equation}
T_1 = \| D^{-1/2}(W - \mcW) D^{-1/2}\|_{(2,2)} \leq  \|D^{-1}\|_{(2,2)} \| W - \mcW \|_{(2,2)}. \nonumber
\end{equation}
Note that 
\begin{eqnarray}
\|\mcD^{-1}\|_{(2,2)} &=& \underset{l}{\mathrm{max}} \left| \eig_l(\mcD^{-1})\right|  =  \underset{l}{\mathrm{max}} \frac{1}{\left(\sum_{j \neq l} w_{lj}^0\right)} \nonumber \\
&=& \frac{1}{\underset{l}{\mathrm{min}} \left(\sum_{j \neq l} w_{lj}^0\right)}  = \frac{1}{p\tau_p},
\end{eqnarray}
where $\tau_p$ is the minimum degree of a vertex divided by the maximum possible degree.
We have assumed that $\tau_p^2 > (p \log p)^{-1}$, so that $p\tau_p > p^{1/2}(\log p)^{-1/2}$. Hence, we obtain $\|\mcD^{-1}\|_{(2,2)} \leq (\log p /p )^{1/2}$.
Also,
\begin{eqnarray*}
\|D^{-1}\|_{(2,2)} & \leq & \|\mcD^{-1}\|_{(2,2)} + \|D^{-1} - \mcD^{-1}\|_{(2,2)} \\
&\leq&  \|\mcD^{-1}\|_{(2,2)} + \|D^{-1}\|_{(2,2)}\|D- \mcD\|_{(2,2)}\|\mcD^{-1}\|_{(2,2)},
\end{eqnarray*}
which gives us
%$\|D^{-1}\|_{(2,2)}\left(1 - \|\mcD^{-1}\|_{(2,2)}\|D - \mcD\|_{(2,2)}\right) \leq \|\mcD^{-1}\|_{(2,2)}$, so that 
\begin{equation}
\|D^{-1}\|_{(2,2)} \leq \frac{\|\mcD^{-1}\|_{(2,2)}}{1 - \|\mcD^{-1}\|_{(2,2)}\|D - \mcD\|_{(2,2)}}.
\end{equation}

Then, 
\begin{eqnarray*}
\| D - \mcD \|_{(2,2)} & \leq & \| D - \mcD \|_2 = \left\{ \tr(D - \mcD)^2\right\}^{1/2} \\
&=& \left\{ \sum_{l=1}^p \left(\sum_{j=1}^p w_{lj} - \sum_{j=1}^p w_{lj}^0\right)^2\right\}^{1/2} \leq \left\{ p \sum_{l=1}^p \sum_{j=1}^p (w_{lj} - w_{lj}^0)^2\right\}^{1/2}\\
%&\leq& \left\{ p \sum_{i=1}^p \sum_{j=1}^p (w_{ij} - \omega_{ij})^2\right\}^{1/2} = p^{1/2}\left\{\tr(W - \mcW)^2\right\}^{1/2} \\
&=& p^{1/2}\| W - \mcW \|_2.
\end{eqnarray*}

Thus, $\| D^{-1} \|_{(2,2)} \leq (\log p/p)^{1/2}$ with high probability if $(\log p)^{1/2} \| W - \mcW \|_2 = o_P(1)$. Hence, we obtain, $T_1 \leq p^{-1/2}(\log p)^{1/2}\| W - \mcW \|_2.$

Now, we determine bounds for the term $T_2$. We have
\begin{eqnarray}
T_2 &=& \| D^{-1/2} \mcW (D^{-1/2} - \mcD^{-1/2})\|_{(2,2)} \nonumber \\
&\leq& \| D^{-1/2}\|_{(2,2)} \| \mcW\|_{(2,2)}\| D^{-1/2} - \mcD^{-1/2} \|_{(2,2)} \nonumber \\
&\leq& \| D^{-1/2}\|_{(2,2)}^2 \| \mcD^{-1/2}\|_{(2,2)} \| \mcW\|_{(2,2)}\| D^{1/2} - \mcD^{1/2} \|_{(2,2)}
\end{eqnarray}
Now, $\| D^{1/2} - \mcD^{1/2} \|_{(2,2)} \leq \| D^{1/2} - \mcD^{1/2} \|_2$, so that we get,
$$T_2 \leq \| D^{-1/2}\|_{(2,2)}^2 \| \mcD^{-1/2}\|_{(2,2)} \| \mcW\|_{(2,2)} p^{1/2} \|W - \mcW\|_2.$$
Also, $\| \mcD^{-1/2} \|_{(2,2)}  = \| \mcD^{-1}\|_{(2,2)}^{1/2} < p^{-1/4}(\log p)^{1/4}$, and $\|D^{-1/2}\|_{(2,2)}^2 = \| D^{-1}\|_{(2,2)} < p^{-1/2}(\log p)^{1/2}.$ Hence, we obtain
$$T_2 \leq p^{-1/4} (\log p)^{3/4} \| W - \mcW \|_2 \| \mcW \|_{(2,2)}.$$
Similar arguments lead to the same bounds as above for $T_3$. We have $\| \mcW \|_{(2,2)} \leq \| \mcW \|_{(\infty,\infty)} = \underset{l}{\mathrm{max}} \sum_{j=1}^p | w_{lj}^0| < p^{\kappa}$, say. Then,
$$ \| L - \mcL\|_{(2,2)} \leq p^{\kappa - 1/4}(\log p)^{3/4} \| W - \mcW \|_2.$$

This completes the proof.
\end{proof}

\newpage

\begin{center}
{\bf \Large Supplementary}
\end{center}

\beginsupplement

\section{Clustering for multiple datasets}
\label{Sec:Integrative}

Suppose we have $J$ data sources with different set of subjects but same set of variables corresponding to each of the subjects and data sources. Our aim is to cluster the variables (for example, proteins in cancer data) for multiple data sources, involving the same set of variables, taking into account the heterogeneity of the data sets and the interdependence between them. We approach this problem from a Bayesian graph clustering point of view as in the main paper, but now modifying our method so as to deal with multiple data sets. We discuss the clustering model and method in the following sections.

\subsection{Clustering model and method}

\cite{lock2013bayesian} considered the problem of clustering a fixed set of objects based on multiple datasets. Their method, called Bayesian Consensus Clustering (BCC), tries to simultaneously determine the source-specific clustering and the global clustering of objects incorporating the heterogeneity of the individual datasets. BCC uses Dirichlet mixture model for different data sources. The conditional distribution of the source-specific clusterings given the global clusterings are modeled through a suitably defined dependence function with an adherence parameter. The number of local and global clusters are taken to be equal. Effective estimation of the clustering model is carried out using a Bayesian framework involving MCMC methods.

We borrow a similar idea like above, but instead of clustering the objects, we cluster the variables in question. We first construct the weighted adjacency matrix corresponding to each of the data sets using the methods described in the Section 3 of the main manuscript, and then obtain the corresponding graph Laplacians. Corresponding to the $J$ data sources, we perform Laplacian embedding of the variables so as to get $\mathcal{Y}_1,\ldots, \mathcal{Y}_J$ as the modified data. The embedded data point for variable $t$ in data set $j$ is denoted by $Y_{jt}$. We assume that there are separate local clusters for each of the data sets (denoted by $\mathcal{C}_j^l$ for data set $j$) and one global cluster involving all the data sets (denoted by $\mathcal{C}^g$). Thus, each $Y_{jt}$ belongs to a local cluster $\mathcal{C}_{jt}^l \in \{1,\ldots,K\}$ and one global cluster $\mathcal{C}_t^g$. Each $Y_{jt},\, t = 1,\ldots,p$ is assumed to be independently drawn from a $K$-component mixture distribution with parameters $\theta_{j1},\ldots,\theta_{jK}$. The local clusters $\mathcal{C}_j^l = (\mathcal{C}_{j1}^l,\ldots,\mathcal{C}_{jp}^l)$ are dependent on the global clusters $\mathcal{C}^g = (\mathcal{C}_1^g,\ldots,\mathcal{C}_p^g)$ as
\begin{equation}
\mathrm{P}(\mathcal{C}_{jt}^l  = k \mid \mathcal{C}_{j}^g) = \nu(k, \mathcal{C}_{j}^g ,\alpha_j),
\end{equation}
where $\nu{\cdot}$ is the dependence function controlled by the parameter $\alpha_j$. The form of the dependence function is given by
\begin{equation}
\nu(\mathcal{C}_{jt}^l , \mathcal{C}_{t}^g , \alpha_j) = \left\{
	\begin{array}{ll}
		\alpha_j  & \mbox{if }  \mathcal{C}_{jt}^l = \mathcal{C}_{t}^g  \\
		\frac{1-\alpha_j}{K-1} & \mbox{otherwise}
	\end{array}
\right.
\end{equation}
where $\alpha_j \in [1/K,1]$ acts as an adherence parameter of the local clustering for data source $j$ to the global one.

Assuming a Dirichlet prior on the probabilities $\pi_k = \mathrm{P}(\mathcal{C}_{t}^g = k)$, we can get the probability that variable $t$ belongs to cluster $k$ in data source $j$ as
\begin{equation}
\mathrm{P}(\mathcal{C}_{jt}^l = k \mid \pi_1,\ldots,\pi_K) = \pi_k \alpha_j + (1 - \pi_k) \frac{1-\alpha_j}{K-1}, 
\end{equation}
and the conditional distribution of the global clusters is obtained as 
\begin{equation}
\mathrm{P}(\mathcal{C}_t^g = k \mid \mathcal{C}_1^l,\ldots,\mathcal{C}_p^l, \pi_1,\ldots,\pi_K , \alpha_1,\ldots,\alpha_J) \propto \pi_k \prod_{j = 1}^{J} \nu(\mathcal{C}_{jt}^l, k, \alpha_j).
\end{equation}

The parameter estimates along with the local and global clusters can be obtained using MCMC as described in \cite{lock2013bayesian}. One major challenge in application of this method is in the choice of the number of clusters. One adhoc technique would be to choose the maximum of the number of clusters obtained from clustering the data sources separately using our method based on DP-means. 

\subsection{Application to pan-Cancer proteomic data}

We consider the proteomics data described in the previous section, and apply the modified BCC method as described above to cluster the proteins integrating over multiple cancer types. We would like to explore how the proteins group together across multiple sources especially when the inter-dependence among the sources are modeled statistically. We consider two different types of lung cancers, namely LUAD and LUSC, to find the groups of proteins which move together in these cancer types. For both the cancer types, we found that for separate clusterings using our method described in this paper, there are nine clusters which have at least four members. Hence we specified the number of clusters in the modified BCC method to be nine. The resulting global cluster is shown in Figure \ref{fig:lungBCC}.

\begin{figure}[h]
\includegraphics[scale=0.25]{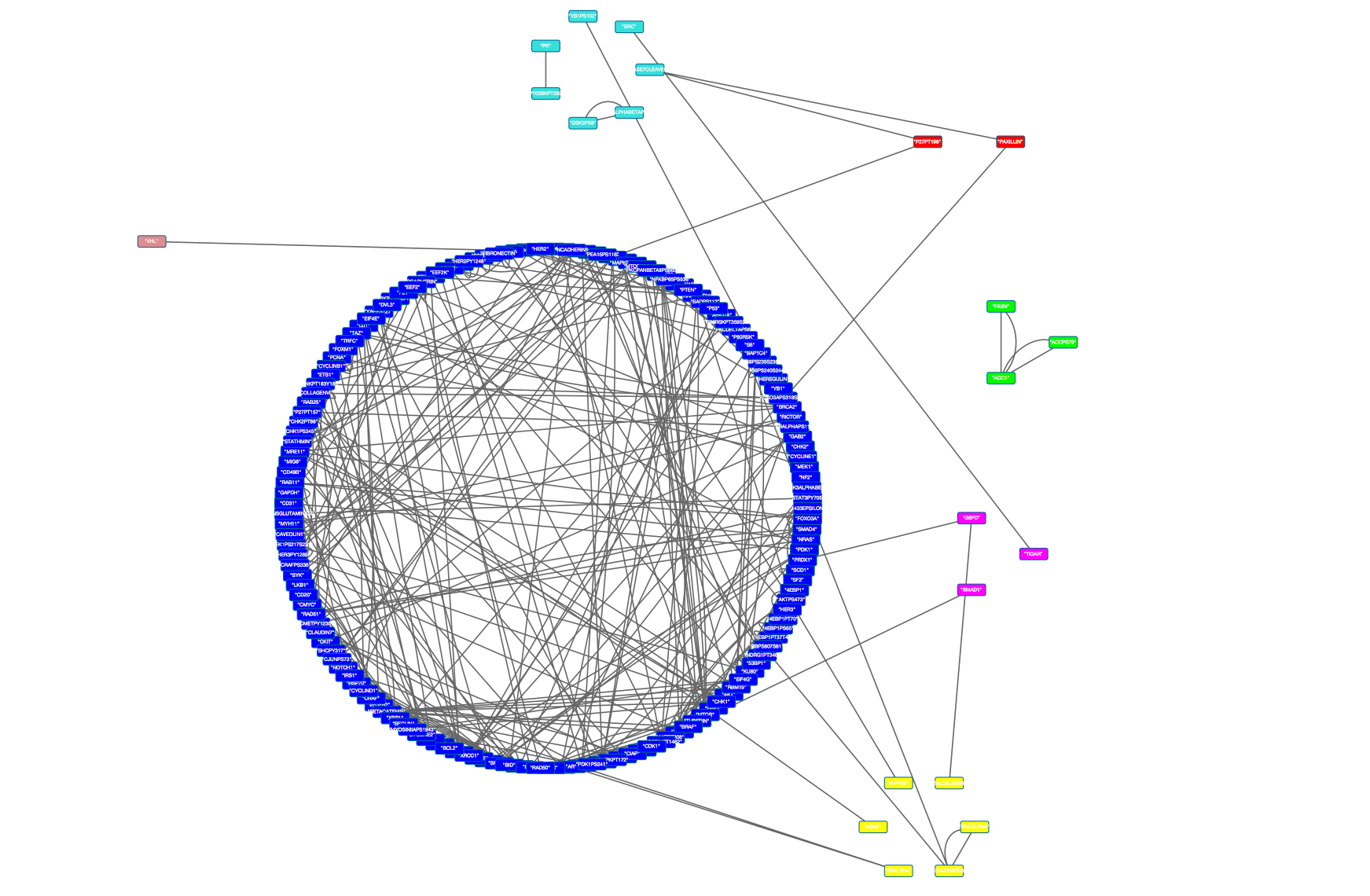}
\caption{Figure showing global clustering of proteins for LUAD and LUSC cancer types.}
\label{fig:lungBCC}
\end{figure}

\newpage
\section{Details of posterior computations using MCMC}

\subsection{MCMC for Bayesian neighborhood selection}

We present the details of the MCMC procedure for the rescaled spike and slab prior. The MCMC procedure is given in detail in \cite{ishwaran2005spike}. The procedure, known as Stochastic Variable Selection (SVS) is given as below. We present the details for the regression model with response $X_l$ and the rest of the variables as regressors.

Posterior values $(\bm{\beta_l}, \bgamma_l, \bm{\tau_l}, u_l, \sigma_l^2 \mid \bm{X_l^*})$ are drawn using the hierarchical Bayesian model (3.2) of the main manuscript, where $\bm{X_l^*} = (X_{l,1}^*,\ldots,X_{l,n}^*)^T$, and $\bm{\tau_l} = (\tau_{l(j)})_{j \neq l}$. Also let us denote $\bX_{-l} = (\bX_1,\ldots,\bX_{l-1},\bX_{l+1},\ldots,\bX_p), \lambda_{l(j)} = \gamma_{l(j)} \tau_{l(j)}^2$ and $\bm{\lambda_l} = (\lambda_{l(j)})_{j \neq l}$. The Gibbs sampler algorithm runs as follows --
\begin{enumerate}
\item Simulate $\bm{\beta_l} \mid \bm{\lambda_l},  \sigma_l^2, \bm{X_l^*}$ from a multivariate Normal distribution with mean $\Sigma_l \bX_{-l}^T \bm{X_l^*}$ and covariance $\sigma_l^2 \Sigma_l$, where 
\begin{equation*}
\Sigma_l = (\bX_{-l}^T\bX_{-l} + n\sigma_l^2 \Lambda_l^{-1})^{-1},\, \Lambda_l = \mathrm{diag}(\lambda_{l(1)},\ldots,\lambda_{l(l-1)}, \lambda_{l(l+1)}, \ldots,\lambda_{l(p)}).
\end{equation*}

\item Simulate 
\begin{equation*}
\gamma_{l(j)} \mid \bm{\beta_l},\bm{\tau_l}, u_l \stackrel{indep.}{\sim} \frac{u_{1,l(j)}}{u_{1,l(j)} + u_{2,l(j)}}\delta_{\nu_0}(\cdot) + \frac{u_{2,l(j)}}{u_{1,l(j)} + u_{2,l(j)}}\delta_1(\cdot), \, j \neq l,
\end{equation*}
where 
$u_{1,l(j)} = (1-u_l)\nu_0^{-1/2}\exp(-\frac{\beta_{l(j)}^2}{2\nu_0\tau_{l(j)}^2}),\, u_{2,l(j)} = u \exp(-\frac{\beta_{l(j)}^2}{2\tau_{l(j)}^2})$.

\item Simulate $\tau_{l(j)}^{-2} \mid \bm{\beta_l}, \gamma_{l(j)} \stackrel{indep.}{\sim} \mathrm{Gamma}\left(a_1 + 1/2, a_2 + \frac{\beta_{l(j)}^2}{2\gamma_{l(j)}}\right),\, j \neq l.$

\item Simulate $u_l \mid  \bm{\lambda_l} \sim \mathrm{Beta}(1 + \# \{j:\lambda_{l(j)} = 1\}, 1 + \# \{ j: \lambda_{l(j)} = \nu_0\})$.

\item Simulate $\sigma_l^{-2} \mid \bm{\beta_l}, \bm{X_l^*} \sim \mathrm{Gamma}(b_1 + n/2, b_2 + \|\bm{X_l^*} - \bX_{-l}\bm{\beta}\|^2/2n).$

\item Set $\lambda_{l(j)} = \gamma_{l(j)} \tau_{l(j)}^2,\, j \neq l.$

\end{enumerate}

The above steps complete one iteration of the Gibbs sampler. 

\subsection{MCMC for Clustering using Dirichlet Process Mixture Models}

We now present the details of the MCMC procedure for the Dirichlet Process Mixture Models used in graph clustering as described in Section 4 of the main manuscript. The choice of conjugate priors leads to availability of exact analytical forms of the posterior distributions so that Gibbs sampling is easily accomplished. The steps in the MCMC procedure are as below.

Given $\{\mu_c^{(t-1)}\}_{c=1}^{C}$ and $\{z_l^{(t-1)}\}_{l=1}^{p}$ from iteration $t-1$, $\{\mu_c^{(t)}\}_{c=1}^{C}$ and $\{z_l^{(t)}\}_{i=1}^{p}$ are sampled as --
\begin{enumerate}
\item Set $\bm{z} = \bm{z}^{(t-1)}$.
\item For $l=1,\ldots,p,$
\begin{enumerate}
\item As we are going to sample a new $z_l$ for data point $\bm{y}_l$, remove $\bm{y}_l$ from cluster $z_l$.
\item If the $\bm{y}_l$ removed above was the only data point in that cluster, the cluster becomes empty. We remove this cluster, and $C$ is decreased by $1$.
\item We re-arrange the clusters so that the cluster labels are $1,\ldots,C.$
\item A new sample for $z_l$ is drawn as
\begin{eqnarray*}
\mathrm{Pr}(z_l = c, c \leq C) &\propto& \frac{n_{c,-l}}{n+\alpha_0-1}\mathrm{N}(\bm{y}_l \mid \mu_c^{(t-1)}), \\
\mathrm{Pr}(z_l = C + 1) &\propto & \frac{\alpha_0}{n+\alpha_0-1}.
\end{eqnarray*}
\item If $z_l = C+1$, we get a new cluster. The cluster index is $C+1$. Sample a new cluster parameter $\mu_{C+1}$ from a multivariate Normal distribution with mean $(\rho^{-1}I + \sigma^{-1}I)^{-1}\sigma^{-1}\bm{y}_l$ and covariance $(\rho^{-1}I + \sigma^{-1}I)^{-1}$.
\end{enumerate}
\item For $c=1,\ldots,C$, sample cluster parameter $\mu_c^{(t)}$ for each cluster from a multivariate Normal distribution $\mathrm{N}((\rho^{-1}I + n_c\sigma^{-1}I)^{-1}\sum_{z_l = c}\bm{y}_l,  (\rho^{-1}I + n_c\sigma^{-1}I)^{-1})$, where $n_c$ is the number of data points in cluster $c$.
\item Set $\bm{z}^{(t)} = \bm{z}$.
\end{enumerate}

\newpage
\section{Supplementary figures for different cancer types}

\begin{figure}[h]
\includegraphics[scale=0.3]{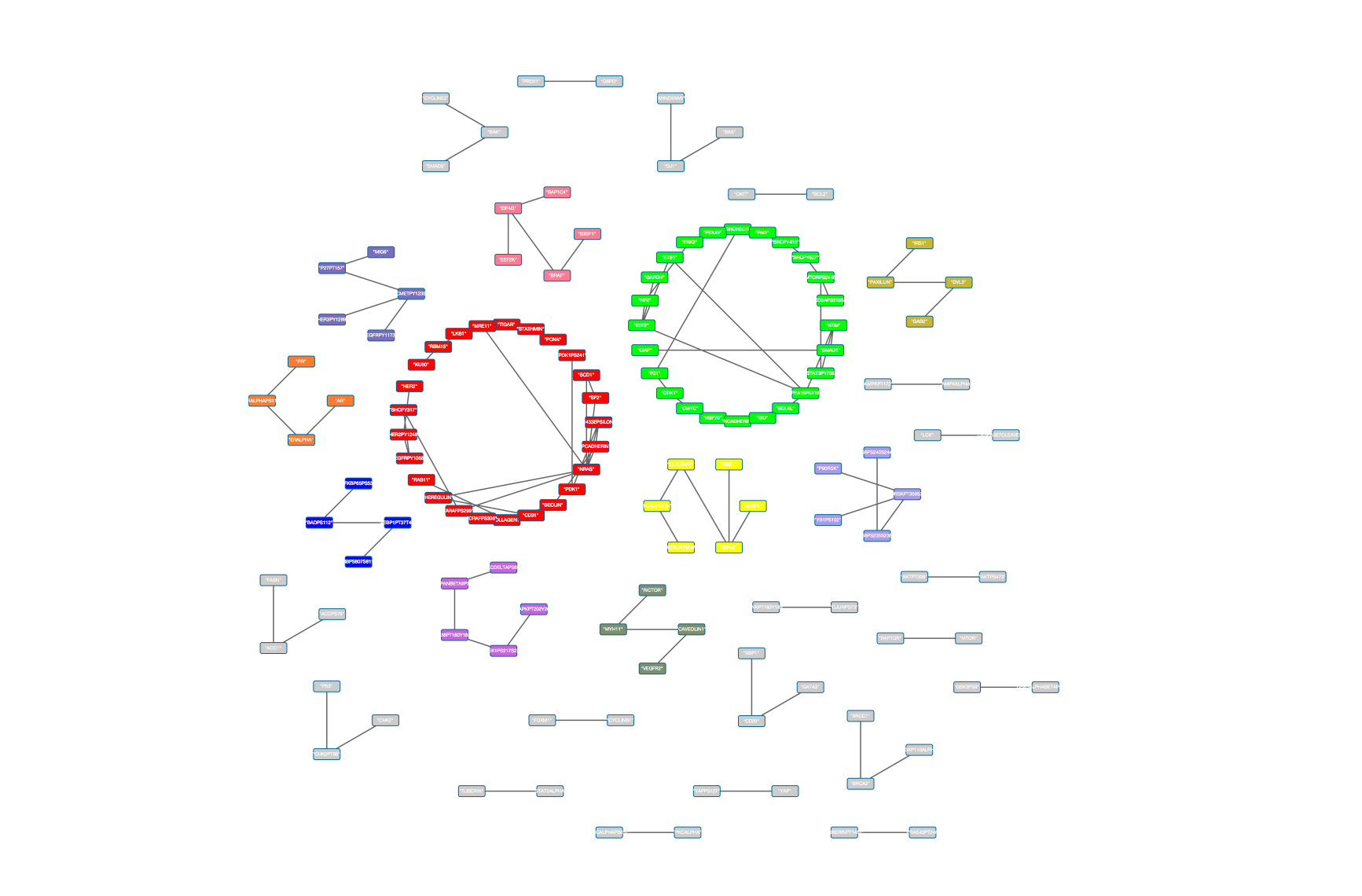}
\caption{Clusters for Uterine corpus endometrial carcinoma (UCEC). Clusters with at least 4 proteins are color-coded; those with fewer are gray.}
\end{figure}

\begin{figure}
\centering
   \begin{subfigure}[b]{0.75\textwidth}
   \includegraphics[width=1\linewidth]{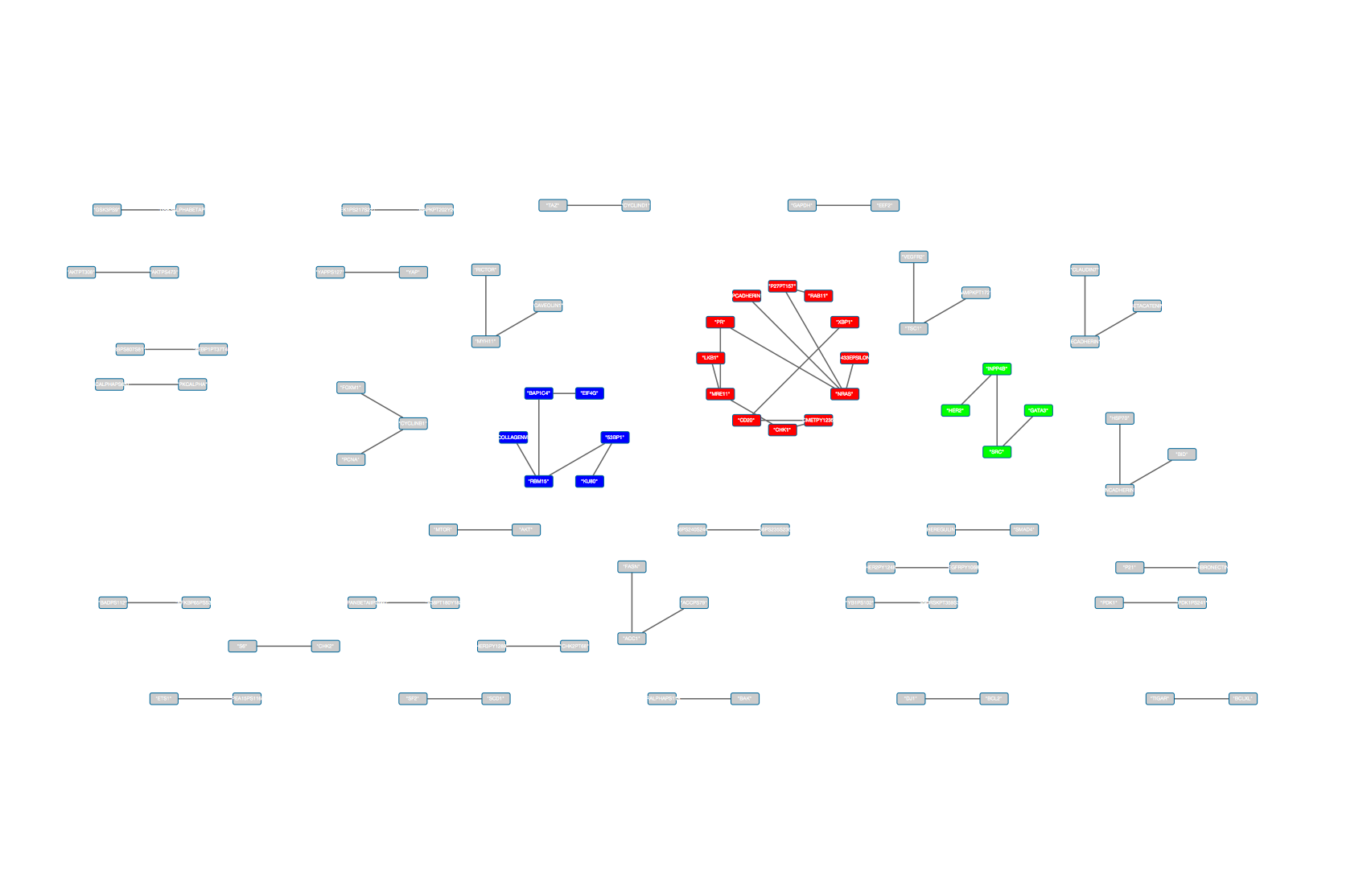}
   \caption{}
\end{subfigure}

\begin{subfigure}[b]{0.75\textwidth}
   \includegraphics[width=1\linewidth]{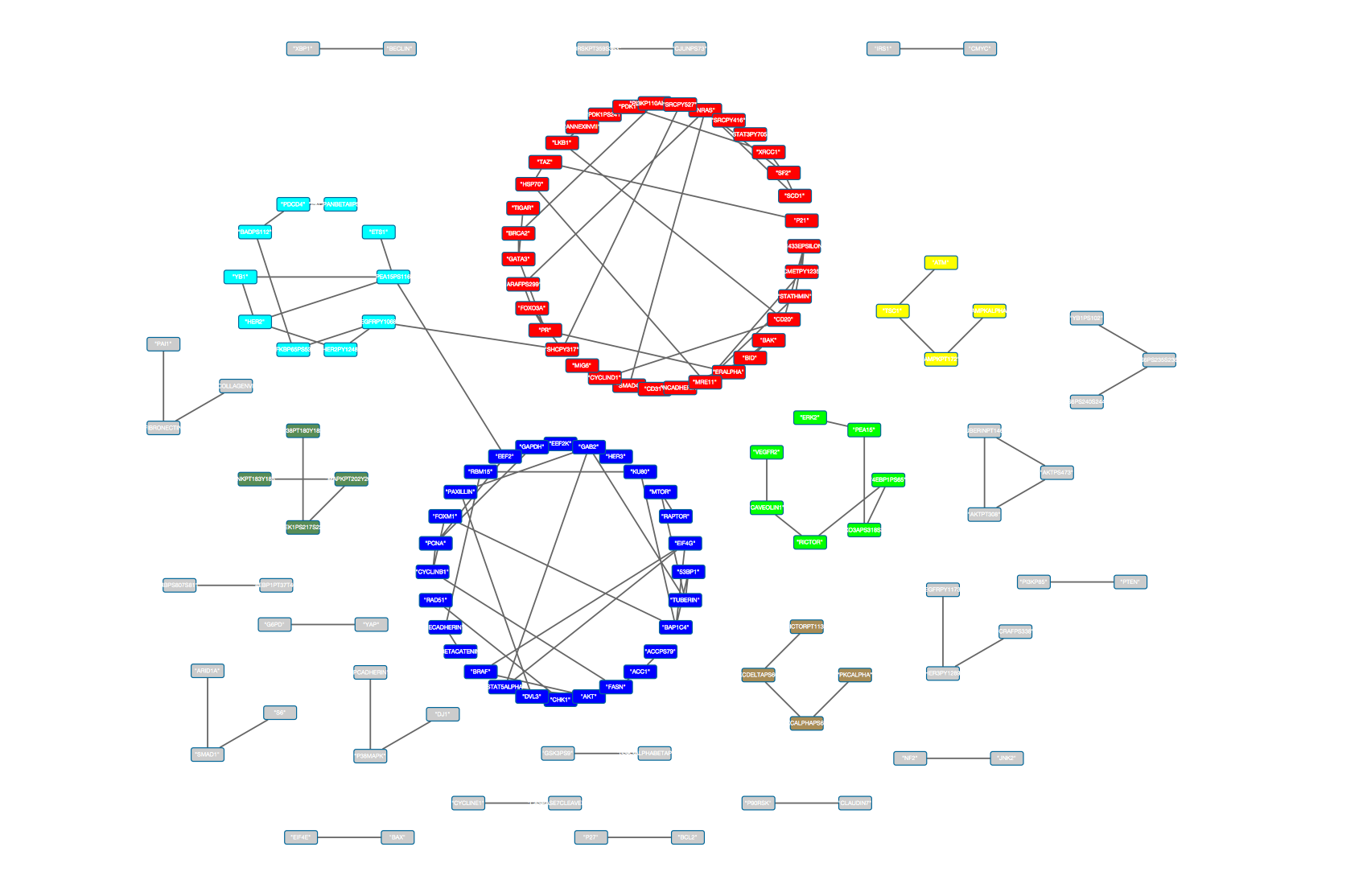}
   \caption{}
\end{subfigure}
\caption{Protein networks for (a) Bladder urothelial carcinoma (BLCA) and (b) Colon adenocarcinoma (COAD) cancer types. Clusters with at least 4 proteins are color-coded; those with fewer are gray.}
\end{figure}

\begin{figure}
\centering
   \begin{subfigure}[b]{0.75\textwidth}
   \includegraphics[width=1\linewidth]{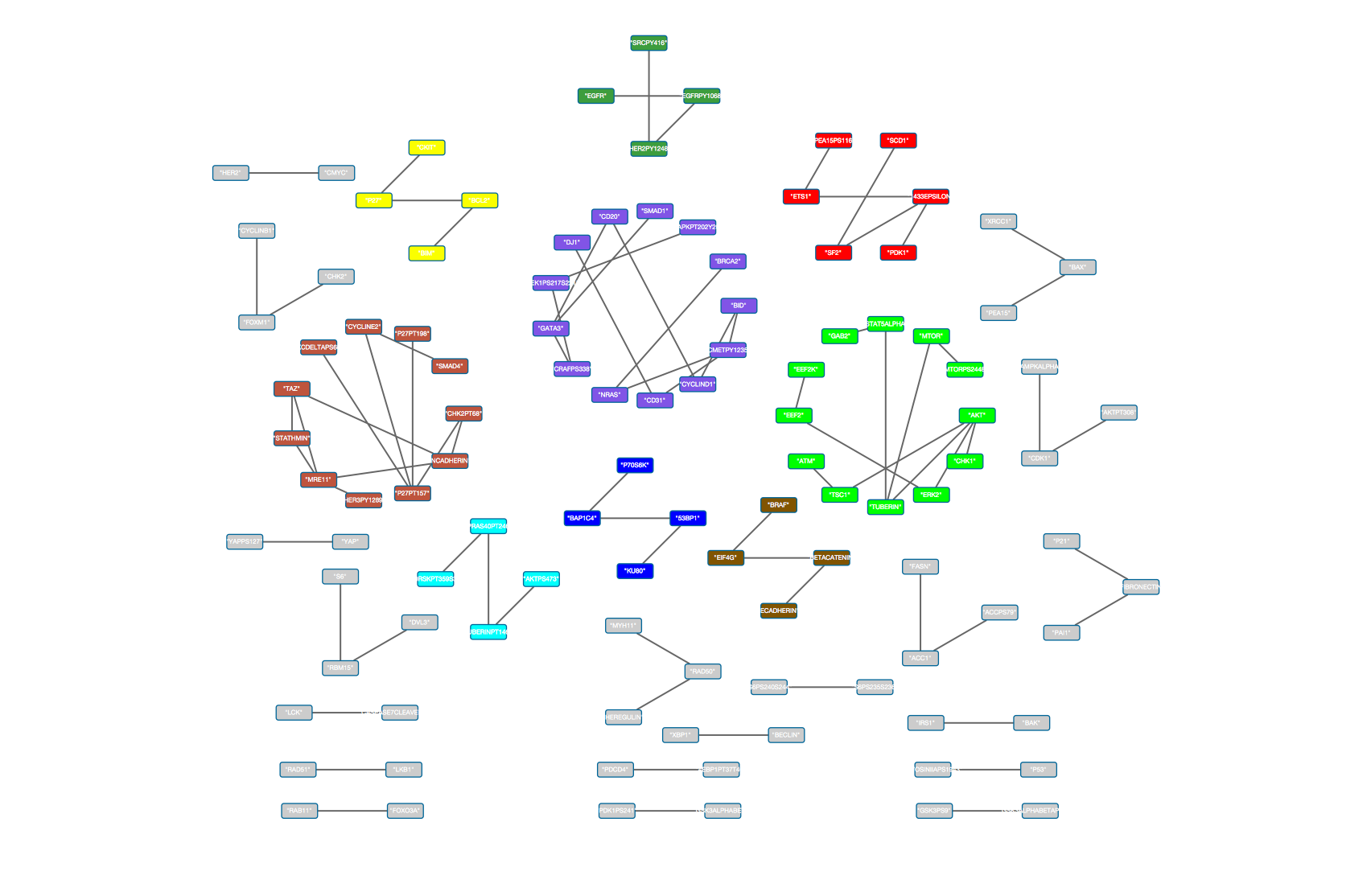}
   \caption{}
\end{subfigure}

\begin{subfigure}[b]{0.75\textwidth}
   \includegraphics[width=1\linewidth]{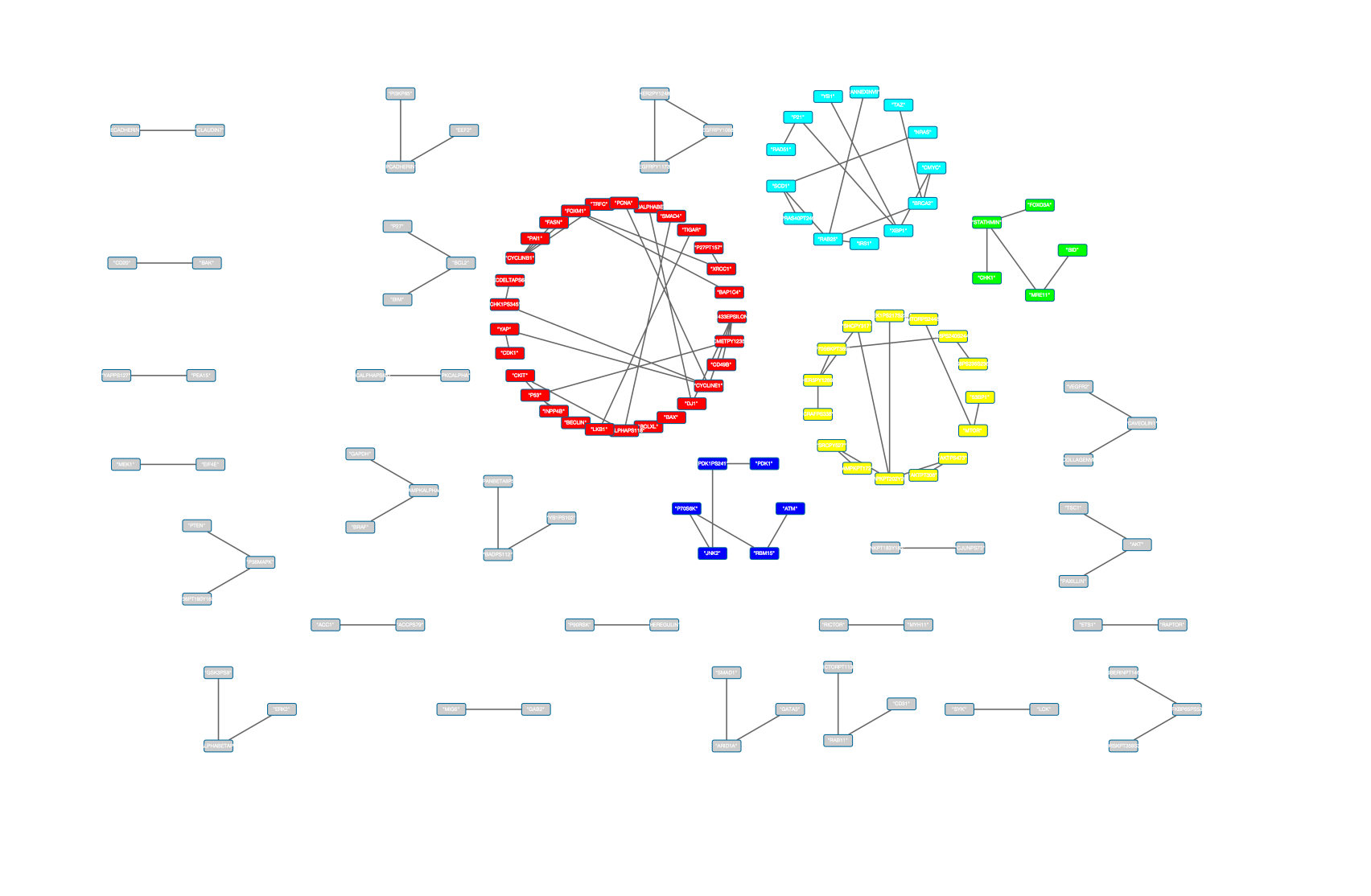}
   \caption{}
\end{subfigure}
\caption{Protein networks for (a) Head and Neck squamous cell carcinoma (HNSC) and (b)  Renal clear cell carcinoma (KIRC) cancer types. Clusters with at least 4 proteins are color-coded, those with fewer are gray.}
\end{figure}

\begin{figure}
\centering
   \begin{subfigure}[b]{0.75\textwidth}
   \includegraphics[width=1\linewidth]{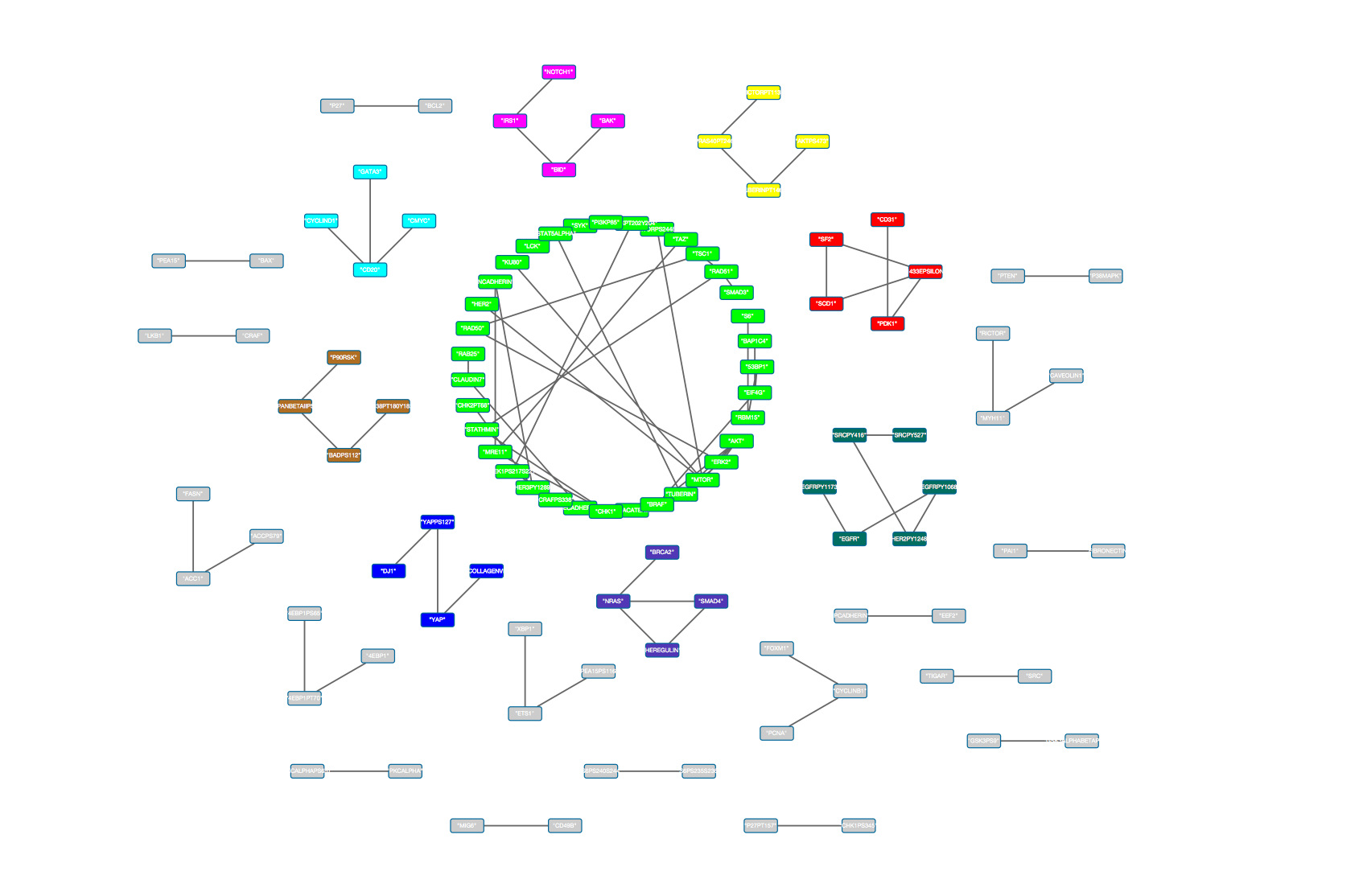}
   \caption{}
\end{subfigure}

\begin{subfigure}[b]{0.75\textwidth}
   \includegraphics[width=1\linewidth]{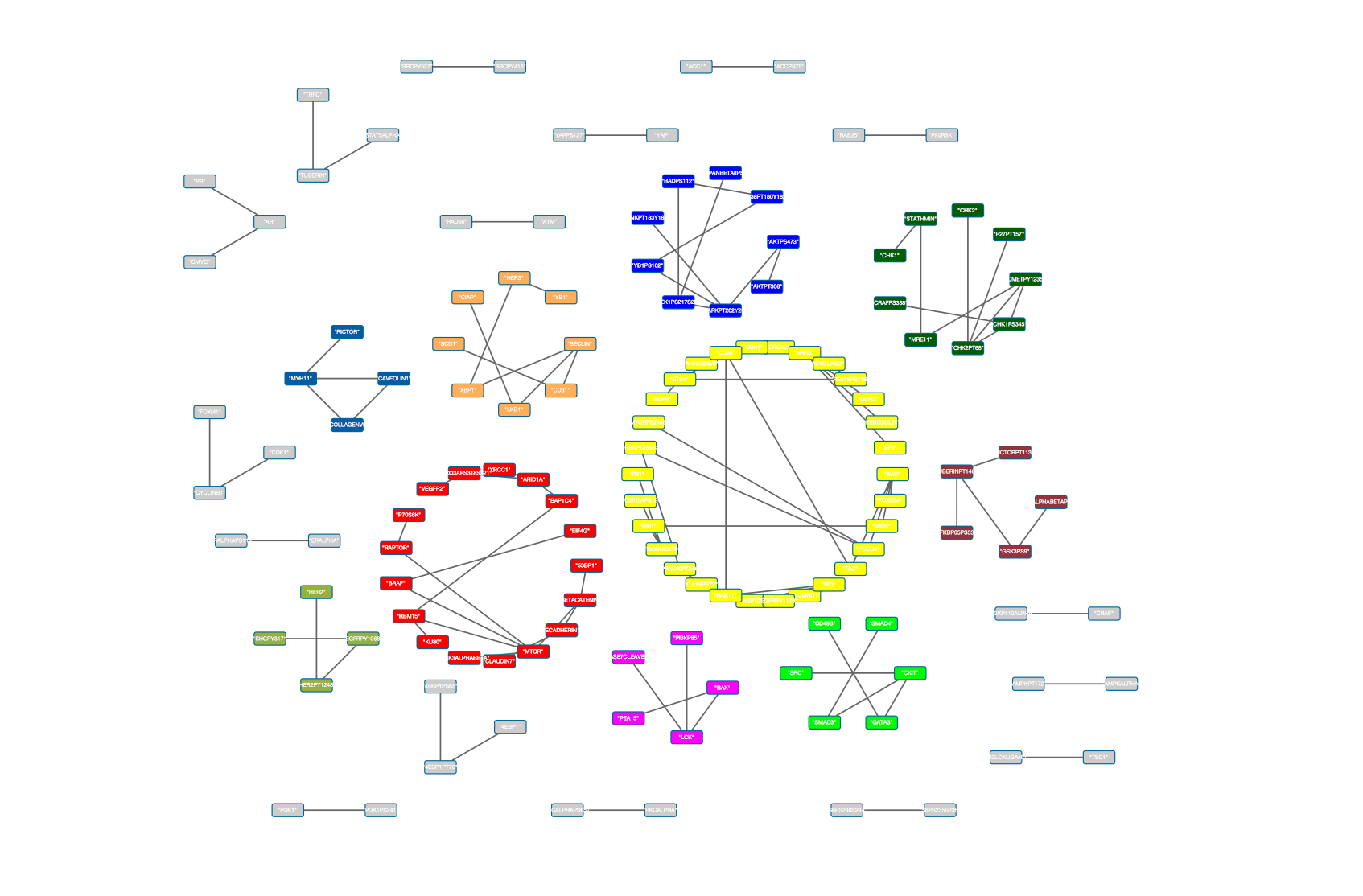}
   \caption{}
\end{subfigure}
\caption{Protein networks for (a) Lung adenocarcinoma (LUAD) and (b) Ovarian cystadenocarcinoma (OVCA) cancer types. Clusters with at least 4 proteins are color-coded, those with fewer are gray.}
\end{figure}

\newpage

\bibliographystyle{apalike}
\bibliography{clustbibfile}

\end{document}